\documentclass[a4paper]{article}
\usepackage[american]{babel}
\usepackage{amsmath,amssymb,amsthm}
\usepackage{graphicx,color}

\renewcommand{\text}{\textnormal}
\newenvironment{remark}{\vskip 0.1cm \noindent\textit{Remark}.}{\vskip 0.1cm}
\newtheorem{theorem}{Theorem}
\newtheorem{corollary}[theorem]{Corollary}
\newcommand{\ave}{\text{ave}}
\newcommand{\B}{\mathcal{B}}
\newcommand{\des}{\text{des}}
\newcommand{\I}{\mathcal{I}}
\newcommand{\ie}{i.e.}
\newcommand{\ind}[1]{{\chi}_{#1}}
\newcommand{\Leb}{\mathcal{L}}
\newcommand{\R}{\mathbb{R}}
\newcommand{\V}{\mathcal{V}_\text{adm}}
\newcommand{\Z}{\mathbb{Z}}
\newcommand{\lint}{\int\limits}
\DeclareMathOperator*{\supp}{supp}
\DeclareMathOperator*{\esssup}{ess\,sup}

\graphicspath{{./}{figure/}}

\title{Multiscale modeling of granular flows with application to crowd dynamics}

\author{Emiliano Cristiani\thanks{CEMSAC, University of Salerno, Salerno, Italy and
        Istituto per le Applicazioni del Calcolo ``M. Picone'', CNR, Roma, Italy
        ({\tt emiliano.cristiani@gmail.com}).}
		\and Benedetto Piccoli\thanks{Istituto per le Applicazioni del Calcolo ``M. Picone'',
		CNR, Roma, Italy ({\tt b.piccoli@iac.cnr.it}).}
		\and Andrea Tosin\thanks{Department of Mathematics, Politecnico di Torino, Torino, Italy
		({\tt andrea.tosin@polito.it}).}
		}
		
\date{}

\begin{document}

\maketitle

\begin{abstract}
In this paper a new multiscale modeling technique is proposed. It relies on a recently introduced measure-theoretic approach, which allows to manage the microscopic and the macroscopic scale under a unique framework. In the resulting coupled model the two scales coexist and share information. This allows to perform numerical simulations in which the trajectories and the density of the particles affect each other. Crowd dynamics is the motivating application throughout the paper.
\end{abstract}



\pagestyle{myheadings}
\thispagestyle{plain}
\markboth{E. CRISTIANI, B. PICCOLI, AND A. TOSIN}{MULTISCALE MODELING OF GRANULAR FLOWS}

\section{Introduction}
\label{sect.intro}
Modeling group dynamics of living systems, such as groups of animals or human crowds, is a difficult task because one can only partially rely on the well-established theories of classical mechanics. Mathematical models must take into account several features of living matter: for example, individuals are not passively dragged by external forces, instead they have a decision-based dynamics; they experience nonlocal interactions, since they are able to see even far group mates and make decisions consequently; interactions can be metric (\ie, with group mates less than a threshold apart) or topological (\ie, with a fixed number of group mates no matter how far they are) \cite{ballerini2008ira}; interactions are strongly anisotropic because the subjects have a limited visual field, and mechanisms for collision avoidance are expected to be mainly directed toward group mates in front \cite{cristiani2010eai}; individuals are different from each other, each of them having for instance her/his own goal, reaction time, and maximal velocity. 

One of the most interesting consequences of these characteristics is the emergence of self-organization. Individuals can deploy themselves to give rise to apparently ordered and coordinated configurations or patterns \cite{KR02}. We cite, for example, clusters by starlings \cite{ballerini2008ira}, lines by elephants, penguins, and lobsters, V-like formations by geese, lanes by pedestrians \cite{helbing2001sop,hoogendoorn2005sop}. Actually such group configurations are not the result of a common decision made by the individuals or by a leader. Instead, they stem from simple rules followed by each individual, which takes into account the position/velocity of a few group mates. It is then possible that a single individual does not even perceive the global structure of the group it is part of.

In this paper we focus on crowd dynamics, which in the last few years has been the object of many mathematical models. In the \textit{microscopic} (\ie, agent-based) approach pedestrians are considered individually. Models usually consist of a (large) system of ordinary differential equations, each of which describes the behavior of a single pedestrian \cite{helbing1995sfm,helbing2001sop,hoogendoorn2003spf,MR2556188,maury2008mfc}.  In the \textit{macroscopic} approach pedestrians are instead described by means of their average density, which in most models obeys conservation or balance laws \cite{MR2438218,MR2158218,MR2438214,henderson1974fmh,hughes2002ctf}.

It is not fair to state that either approach is better for whatever problem. Rather, it is clear that a microscopic approach is advantageous when one wants to model differences among the individuals, random disturbance, or small environments. Moreover, it is the only reliable approach when one wants to track exactly the position of a few walkers. On the other hand, it may not be convenient to use a microscopic approach to model pedestrian flow in large environments, due to the high computational effort required. A macroscopic approach may be preferable to address optimization problems and analytical issues, as well as to handle experimental data. Nonetheless, despite self-organization phenomena are often visible only in large crowds \cite{helbing2007qab}, they are a consequence of strategical behaviors developed by individual pedestrians.

In \cite{cristiani2010mso,piccoli2010tem,piccoli2009pfb} we have extensively analyzed a measure-based modeling framework able to describe group behavior at both the microscopic and the macroscopic scale. The key point is the reinterpretation of the classical conservation laws in terms of abstract mass measures, which are then specialized to singular Dirac measures for microscopic models and to absolutely continuous measures (w.r.t. Lebesgue, \ie, the volume measure) for macroscopic models. We have shown \cite{cristiani2010mso} that the two scales may reproduce the same features of the group behavior, thus providing a perfect matching between the results of the simulations for the microscopic and the macroscopic model in some test cases. This motivated the multiscale approach that we propose here. Such an approach allows to keep a macroscopic view without losing the right amount of ``granularity'', which is crucial for the emergence of some self-organized patterns. Furthermore, the proposed method allows to introduce in a macroscopic (averaged) context some microscopic effects, such as random disturbances or differences among the individuals, in a fully justifiable manner from both the physical and the mathematical perspective. In the model that we propose, microscopic and macroscopic scales coexist and continuously share information on the overall dynamics. More precisely, the microscopic and the macroscopic part of the model track the trajectories of single pedestrians and the density of pedestrians, respectively, using the same evolution equation duly interpreted in the sense of measures. In this respect, the two scales are indivisible. This makes the difference from other ways of understanding multiscale approaches in the literature. For example, in \cite{quarteroni2003agm} a multiscale geometric technique is used to represent the circulatory system: one specific part of the network is accurately modeled in three dimensions, whereas the rest is described by means of lumped zero-dimensional models. This enables one to account for the whole circulatory network while keeping the complexity of the model under control. Multiscale methods can be implemented also at a numerical level in connection with domain decomposition (see e.g., \cite{donev2010hpc,weimar2001cmm}), in order to compute the solution to a certain equation with different local accuracy. The general idea is to couple accurate but expensive calculations, performed by a microscopic (e.g., particle-based) solver in small and inhomogeneous regions, with less accurate but also less expensive ones, performed by a macroscopic (e.g., continuum) solver in large and homogeneous regions. The two solvers usually exchange information at the interface of the respective regions. Another possibility (see e.g., \cite{kraft1997emc}) is to alternate the two scales for computing on the same system. In the resulting iterative algorithm, the output of the microscopic simulation is used as input for the macroscopic simulation and \textit{vice versa}. A further way to understand the multiscale approach is through upscaling procedures. In this case, the ultimate goal is to pass from a detailed but often inhomogeneous description of some quantities to a rougher but more homogeneous representation, by averaging out inhomogeneities via homogenization techniques (see e.g., \cite{bresch2010rie}).

It is worth pointing out that the dichotomy fine vs. coarse scale does not necessarily imply a parallel dichotomy ODE vs. PDE modeling. In other words, it is possible that the underlying mathematical models pertain to the continuum theory at both scales (like in most of the examples recalled above) or that the multiscale coupling between a discrete and a continuous model is realized only at an approximate computational level by averaging and sampling. Conversely, in the multiscale approach we propose here, the microscopic scale is actually a discrete one which complements the continuous flow with granularity. The resulting model is then a coupled microscopic-macroscopic one, and computational schemes are derived accordingly.

The paper is organized as follows. Section \ref{sect.time.evol.meas} introduces the measure-theoretic framework and models pedestrian kinematics. Section \ref{sect.multiscale} details the multiscale approach, addressing in particular the choice of microscopic and macroscopic parameters and their scaling. Section \ref{sect.discr.time} introduces and qualitatively analyzes a discrete-in-time counterpart of the multiscale model. Section \ref{sect.num.approx} proposes a numerical approximation of the equations, with special emphasis on the discretization in space of the macroscopic scale, and explains in detail the resulting numerical algorithm. Section \ref{sect.num.res} discusses the results of numerical simulations in some case studies aimed at checking the effects of the multiscale coupling on the crowd dynamics predicted by the model. Section \ref{sect.conclusions} finally draws conclusions and briefly sketches research perspectives.

\section{Mathematical modeling by time-evolving measures}
\label{sect.time.evol.meas}
From the mathematical point of view the mass of a $d$-dimensional system ($d=1,\,2,\,3$ for physical purposes) at time $t$ is a Radon positive measure $\mu_t$, that we assume to be defined on the Borel $\sigma$-algebra $\B(\R^d)$. For any $E\in\B(\R^d)$ the number $\mu_t(E)\geq 0$ gives the mass of pedestrians contained in $E$ at time $t\geq 0$. In principle, the only further property satisfied by $\mu_t$ is the $\sigma$-additivity, directly translating the principle of additivity of the mass.

Let $T>0$ denote a certain final time. Following \cite{canuto2008eaa}, the conservation of the mass transported by a velocity field $v=v(t,\,x):[0,\,T]\times\R^d\to\R^d$ is expressed by the equation
\begin{equation}
	\frac{\partial\mu_t}{\partial t}+\nabla\cdot{(\mu_t v)}=0, \quad (x,\,t)\in\R^d\times(0,\,T],
	\label{eq.cons.mass.strong}
\end{equation}
along with some given initial distribution of mass $\mu_0$ (initial condition). Derivatives appearing in Eq. \eqref{eq.cons.mass.strong} are meant in the functional sense of measures. Specifically, for every smooth test function $\phi$ with compact support, \ie, $\phi\in C^1_0(\R^d)$, and for a.e. $t\in[0,\,T]$, it results
\begin{equation}
	\frac{d}{dt}\lint_{\R^d}\phi(x)\,d\mu_t(x)=\lint_{\R^d}v(t,\,x)\cdot\nabla{\phi}(x)\,d\mu_t(x),
	\label{eq.cons.mass.weak}
\end{equation}
where integration-by-parts has been used at the right-hand side. A sufficient condition for Eq. \eqref{eq.cons.mass.weak} to be well-defined is that $v(t,\,\cdot)$ is integrable w.r.t. $\mu_t$ for a.e. $t\in[0,\,T]$.

A family of time-evolving measures $\{\mu_t\}_{t>0}$ is said to be a (weak) solution to Eq. \eqref{eq.cons.mass.strong} if, for all $\phi\in C^1_0(\R^d)$, the mapping $t\mapsto\int_{\R^d}\phi(x)\,d\mu_t(x)$
is absolutely continuous and satisfies Eq. \eqref{eq.cons.mass.weak}. In particular, the latter statement means
\begin{equation}
	\lint_{\R^d}\phi(x)d\mu_{t_2}(x)-\lint_{\R^d}\phi(x)d\mu_{t_1}(x)=
		\lint_{t_1}^{t_2}\lint_{\R^d}v(t,\,x)\cdot\nabla{\phi}(x)\,d\mu_t(x)\,dt
	\label{eq.cons.mass.weak.solution}
\end{equation}
for all $t_1,\,t_2\in[0,\,T]$, $t_1\leq t_2$, and all $\phi\in C^1_0(\R^d)$.

\subsection*{Modeling the interactions among pedestrians}
Equation \eqref{eq.cons.mass.strong} provides the evolution of the measure $\mu_t$ as long as the velocity is specified. In our case, given the absence of a balance of linear momentum, this implies modeling directly the field $v$. For this reason, our approach will result in a \emph{first-order model}.

First-order models are quite common in the literature, especially at the macroscopic scale. The velocity can be either specified as a known function \cite{maury2010mcm} or linked to the density of pedestrians by means of empirical fundamental relations $v=v(\rho)$ \cite{colombo2009ens,hughes2002ctf,venuti2007imp}. Sometimes a functional dependence on the density gradient is envisaged, in order to model the sensitivity of pedestrians to the variations of the surrounding density field \cite{MR2438218,MR2438214}. Microscopic models focus instead more closely on the interactions among pedestrians, normally expressing them in terms of generalized forces. They resort therefore to a classical Newtonian paradigm, in which the acceleration is modeled explicitly \cite{helbing2001trs,helbing1995sfm}. We remark, however, that in \cite{maury2007hcc,maury2008mfc} the authors adopt a kinematic modeling of the interactions in the frame of a microscopic model.

With the aim of setting up a model based on the mass conservation only, but in which the microscopic granularity complements the macroscopic dynamics, we cannot entirely resort either to generalized forces or to fundamental relations. Taking advantage of the mass conservation equation in the form \eqref{eq.cons.mass.weak}, which does not assume \emph{a priori} any modeling scale, our approach will be at the same time kinematic, macroscopic, and focused on the strategy developed by pedestrians at the microscopic scale.

To be more specific, let the velocity be expressed in the following form:
\begin{equation}
	v(t,\,x):=v[\mu_t](x)=v_{\textrm{des}}(x)+\nu[\mu_t](x),
	\label{eq.v}
\end{equation}
the square brackets denoting functional dependence on the measure $\mu_t$.

The function $v_{\textrm{des}}:\R^d\to\R^d$ is the \emph{desired velocity}, \ie, the velocity that pedestrians would set to reach their destination if they did not experience mutual interactions. In the simplest case it is a constant field, whereas in more complicated situations it accounts for the presence of possible obstacles to be bypassed (e.g., pedestrians walking in built environments). In our approach the desired velocity is deduced \emph{a priori} from the geometry of the domain, meaning that it is totally independent of the measure $\mu_t$. In other words, it can be regarded to all purposes as a datum of the problem. It is not restrictive to assume that it has constant modulus:
\begin{equation}
	\vert v_{\textrm{des}}(x)\vert=V, \quad \forall\,x\in\R^d,
	\label{eq.des.speed}
\end{equation}
where $V$ represents some characteristic speed of the walkers. We refer the reader to \cite{piccoli2009pfb} for a possible method to construct $v_{\textrm{des}}$.

The function $\nu[\mu_t]:\R^d\to\R^d$ is the \emph{interaction velocity}, that is, the correction that pedestrians make to their desired velocity in consequence of the interactions. The \emph{non-locality} of the interactions is introduced in this framework by deriving $\nu[\mu_t]$ from a synthesis of the information on the crowd distribution around each pedestrian. Specifically, we assume
\begin{equation}
	\nu[\mu_t](x)=\lint_{\R^d\setminus\{x\}}f(\vert y-x\vert)g(\alpha_{xy})\frac{y-x}{\vert y-x\vert}\,d\mu_t(y),
	\label{eq.nu}
\end{equation}
where:
\begin{itemize}
\item $f:\R_+\to\R$ is a function with compact support describing how the walker in $x$ interacts with her/his neighbors on the basis of their distance. If $\supp{f}=[0,\,R]$ for some $R>0$, then a \emph{neighborhood of interaction} is defined for the point $x$ coinciding with the ball $B_R(x)\subset\R^d$ centered in $x$ with radius $R$;
\item $\alpha_{xy}\in[-\pi,\,\pi]$ is the angle between the vectors $y-x$ and $v_{\textrm{des}}(x)$, that is, the angle under which a point $y$ is seen from $x$ with respect to the desired direction of motion;
\item $g:[-\pi,\,\pi]\to[0,\,1]$ is a function which reproduces the angular focus of the walker in $x$.
\end{itemize}
Integration w.r.t. $\mu_t$ accounts for the mass that the walkers see, considering that two fundamental attitudes characterize pedestrian behavior:
\begin{itemize}
\item \emph{repulsion}, \ie, the tendency to avoid collisions and crowded areas;
\item \emph{attraction}, \ie, the tendency, under some circumstances, to not lose the contact with other group mates  (e.g., groups of tourists in guided tours, groups of people sharing specific relationships such as families or parties).
\end{itemize}

Focusing on one of the simplest choices, nonetheless physiologically sound, we suggest for $f$ the following expression:
\begin{equation}
	f(s)=-\frac{F_r}{s}\ind{[0,\,R_r]}(s)+F_as\ind{[0,\,R_a]}(s),
	\label{eq.f}
\end{equation}
where $F_r,\,F_a>0$ are repulsion and attraction strengths, and $R_r,\,R_a>0$ are repulsion and attraction radii. This form of $f$ translates the basic idea that repulsion and attraction are inversely and directly proportional, respectively, to the distance separating the interacting pedestrians.

As pointed out in the Introduction, interactions can be either metric or topological. An interaction is \emph{metric} if the corresponding radius is fixed, so that each walker interacts with all other pedestrians within that given maximum distance. Conversely, an interaction is \emph{topological} if the corresponding radius is adjusted dynamically by each walker, in such a way that the neighborhood of interaction encompasses a predefined mass of other pedestrians s/he feels comfortable to interact with. In this paper we will be mainly concerned with metric interactions, for both repulsion and attraction. The interested reader is referred to \cite{ballerini2008ira,cristiani2010mso}, and references therein, for a detailed discussion of metric and topological effects, also by means of examples and numerical simulations.

The function $g$ carries the \emph{anisotropy} of the interactions, which essentially consists in that pedestrians cannot see all around them and they are not equally sensitive to external stimuli coming from different directions. If $\bar{\alpha}\in[0,\,\pi]$ is the maximum sensitivity angular width, a very simple form of $g$ is
\begin{equation}
	g(s)=\ind{\{\vert s\vert\leq\bar{\alpha}\}}(s), \quad s\in [-\pi,\,\pi].
	\label{eq.g}
\end{equation}
By mollifying this function it is possible to account for the visual fading that usually occurs laterally in the visual field when approaching the maximum angular width\footnote{In order to differentiate also the maximum angular widths of repulsion and attraction \cite{cristiani2010eai}, one may generalize Eq. \eqref{eq.nu} by replacing the product of $f$ and $g$ with a function of two variables accounting simultaneously for $\vert y-x\vert$ and $\alpha_{xy}$.}.

\section{The multiscale approach}
\label{sect.multiscale}
The framework presented in Section \ref{sect.time.evol.meas} is suitable to obtain, as particular cases, models at both the microscopic and the macroscopic scale. In this section we first briefly review the methodology for their individual derivation, already proposed in \cite{cristiani2010mso} to study microscopic and macroscopic self-organization in animal groups and crowds. Then, exploiting the tools offered by the measure-theoretic setting, we merge these concepts into a unique multiscale model, in which the microscopic and the macroscopic dynamics coexist.

\subsection{Microscopic models}
Let us consider a population of $N$ pedestrians, whose positions at time $t$ are denoted $\{P_j(t)\}_{j=1}^N$. In this case the mass of a set $E\in\B(\R^d)$ is the number of pedestrians contained in $E$, that is:
\begin{equation*}
	\mu_t(E)=\operatorname{card}\{P_j(t)\in E\},
\end{equation*}
hence $\mu_t$ is the \emph{counting measure}. We represent it as a sum of Dirac masses, each centered in one of the $P_j$'s:
\begin{equation*}
	\mu_t=\sum_{j=1}^{N}\delta_{P_j(t)}.
\end{equation*}
Plugging this in Eq. \eqref{eq.cons.mass.weak} gives
\begin{equation}
	\frac{d}{dt}\sum_{j=1}^{N}\phi(P_j(t))=\sum_{j=1}^{N}v(t,\,P_j(t))\cdot\nabla{\phi}(P_j(t)), \quad
		\forall\,\phi\in C^1_0(\R^d),
	\label{eq.micro.weak}
\end{equation}
whence, taking the time derivative at the left-hand side and rearranging the terms,
\begin{equation*}
	\sum_{j=1}^{N}\left[\dot{P}_j(t)-v(t,\,P_j(t))\right]\cdot\nabla{\phi}(P_j(t))=0,
\end{equation*}
the dot over $P_j$ standing for derivative w.r.t. $t$. The arbitrariness of $\phi$ implies
\begin{equation}
	\dot{P}_j(t)=v[\mu_t](P_j(t)), \quad j=1,\,\dots,\,N,
	\label{eq.mod.micro}
\end{equation}
where we have set $v(t,\,P_j(t))=v[\mu_t](P_j(t))$ according to Eq. \eqref{eq.v}, therefore the microscopic model specializes in a dynamical system of $N$ coupled ODEs for the $P_j$'s. The coupling is realized by the measure $\mu_t$ in the velocity field. In particular, the microscopic counterpart of Eq. \eqref{eq.nu} reads
\begin{equation*}
	\nu[\mu_t](P_j)=\sum_{\substack{k=1,\,\dots,\,N\\ P_k\ne P_j}}f(\vert P_k-P_j\vert)g(\alpha_{kj})
		\frac{P_k-P_j}{\vert P_k-P_j\vert},
\end{equation*}
where $\alpha_{kj}\in[-\pi,\,\pi]$ is shorthand for the angle formed by the vectors $P_k-P_j$ and $v_\des(P_j)$. We point out that, with the function $f$ given by Eq. \eqref{eq.f}, the statement $P_k\ne P_j$ in the above formula can be converted into the milder one $k\ne j$. Indeed one can prove that if the $P_j$'s are initially all distinct they remain distinct at all successive times $t>0$ (see \cite{cristiani2010eai} for technical details).

\subsection{Macroscopic models}
Macroscopic models are based on the assumption that the matter is continuous, thus the measure $\mu_t$ is absolutely continuous w.r.t. the $d$-dimensional Lebesgue measure $\Leb^d$, $\mu_t\ll\Leb^d$. Radon-Nikodym's Theorem asserts that there exists a function $\rho(t,\,\cdot)\in L^1_{\text{loc}}(\R^d)$ such that
\begin{equation}
	d\mu_t=\rho(t,\,\cdot)\,d\Leb^d, \quad \rho(t,\,\cdot)\geq 0\ \text{a.e.},
	\label{eq.mu.macro}
\end{equation}
called the \emph{density} of $\mu_t$ w.r.t. $\Leb^d$. In our context $\rho(t,\,x)$ represents the density of pedestrians at time $t$ in the point $x$.

Using $\rho$, the mass conservation equation \eqref{eq.cons.mass.weak} rewrites as
\begin{equation}
	\frac{d}{dt}\lint_{\R^d}\rho(t,\,x)\phi(x)\,dx=\lint_{\R^d}\rho(t,\,x)v(t,\,x)\cdot\nabla{\phi(x)}\,dx,
		\quad\forall\,\phi\in C^1_0(\R^d),
	\label{eq.macro.weak}
\end{equation}
namely a weak form of the continuity equation
\begin{equation}
	\frac{\partial\rho}{\partial t}+\nabla\cdot{(\rho v)}=0.
\label{eq.continuity.equation}
\end{equation}
The interaction velocity specializes as
\begin{equation*}
	\nu[\mu_t](x)=\lint_{\R^d}f(\vert y-x\vert)g(\alpha_{xy})\frac{y-x}{\vert y-x\vert}\rho(t,\,y)\,dy
\end{equation*}
where it should be noticed that the domain of integration may now indifferently include or not the point $x$ because $\{x\}$ is a Lebesgue-negligible set.

\subsection{Multiscale models}
If the measure $\mu_t$ is neither purely atomic nor entirely absolutely continuous w.r.t. $\Leb^d$ but includes both parts, we get models that incorporate the microscopic granularity of pedestrians in the macroscopic description of the crowd flow. More specifically, we consider
\begin{equation}
	\mu_t=\theta m_t+(1-\theta)M_t,
	\label{eq.mu.multi}
\end{equation}
where
\begin{equation*}
	m_t=\sum_{j=1}^N\delta_{P_j(t)}, \quad d M_t(x)=\rho(t,\,x)\,dx
\end{equation*}
are the microscopic and the macroscopic mass, respectively. The parameter $\theta\in[0,\,1]$ weights the coupling between the two scales, from $\theta=0$ corresponding to a purely macroscopic model to $\theta=1$ corresponding to a purely microscopic model. In Eq. \eqref{eq.mu.multi} no scaling parameters explicitly appear, but we anticipate that they will arise naturally from our next dimensional analysis (cf. Section \ref{sect.dim.anal}).

Using the measure \eqref{eq.mu.multi}, the mass conservation equation \eqref{eq.cons.mass.weak} takes the form of a mix of microscopic and macroscopic contributions:
\begin{align*}
	\frac{d}{dt}\Biggl(&\theta\sum_{j=1}^{N}\phi(P_j(t))+(1-\theta)\lint_{\R^d}\rho(t,\,x)\phi(x)\,dx\Biggr)=\\
		&\!\!\theta\sum_{j=1}^{N}v(t,\,P_j(t))\cdot\nabla{\phi}(P_j(t))+
			(1-\theta)\lint_{\R^d}\rho(t,\,x)v(t,\,x)\cdot\nabla{\phi}(x)\,dx,
				\quad\forall\,\phi\in C^1_0(\R^d),
\end{align*}
formally a convex linear combination of Eqs. \eqref{eq.micro.weak}, \eqref{eq.macro.weak}. The interaction velocity $\nu[\mu_t]$ is now given by
\begin{align*}
	\nu[\mu_t](x)&=\theta\sum_{\substack{k=1,\,\dots,\,N\\ P_k(t)\ne x}}f(\vert P_k(t)-x\vert)g(\alpha_{xP_k(t)})
		\frac{P_k(t)-x}{\vert P_k(t)-x\vert}\\
	&\phantom{=}+(1-\theta)\lint_{\R^d}f(\vert y-x\vert)g(\alpha_{xy})\frac{y-x}{\vert y-x\vert}
		\rho(t,\,y)\,dy,
\end{align*}
therefore it coincides neither with the fully microscopic nor with the fully macroscopic one. This definitely makes the overall dynamics not a simple superposition of the individual microscopic and macroscopic dynamics.

It is worth noticing that the point $x$ may or may not be one of the positions of the microscopic pedestrians. Computing $\nu[\mu_t]$ for $x=P_j(t)$ shows that the interaction velocity of the $j$-th pedestrian accounts not only for other microscopic pedestrians contained in the neighborhood of interaction but also for the macroscopic density distributed therein, which represents some crowd whose subjects are not individually modeled. Specifically, the term responsible for this is
\begin{equation}
	\lint_{\R^d}f(\vert y-P_j(t)\vert)g(\alpha_{P_j(t)y})\frac{y-P_j(t)}{\vert y-P_j(t)\vert}\rho(t,\,y)\,dy,
	\label{eq.nu.macro-for-micro}
\end{equation}
that we may regard as the macroscopic contribution to the microscopic dynamics. Analogously, computing $\nu[\mu_t]$ for $x$ different from all of the $P_j$'s shows that the interaction velocity of an infinitesimal reference volume centered in $x$ depends not only on the density distributed in the neighborhood of interaction but also on the microscopic pedestrians therein, which play the role of singularities in the average crowd distribution due to the granularity of the flow. The corresponding term is
\begin{equation}
	\sum_{\substack{k=1,\,\dots,\,N\\ P_k(t)\ne x}}f(\vert P_k(t)-x\vert)g(\alpha_{xP_k(t)})
		\frac{P_k(t)-x}{\vert P_k(t)-x\vert},
	\label{eq.nu.micro-for-macro}
\end{equation}
which gives the microscopic contribution to the macroscopic dynamics.

\subsection{Dimensional analysis}
\label{sect.dim.anal}
In order to scale correctly the microscopic and the macroscopic contributions, it is convenient to refer to the non-dimensional form of the model. For this, let us preliminarily notice that the main quantities involved in the equations have the following dimensions:
\begin{itemize}
\item $[t]$ = time
\item $[x]$ = length
\item $[v_\des]=[\nu]$ = length/time
\item $[f]$ = length/(time $\times$ pedestrians)
\item $[\mu_t]$ = pedestrians
\item $[\rho]$ = pedestrians/length$^d$
\end{itemize}
where ``pedestrians'' is actually a dimensionless unit. Additionally, $g$ and $\theta$ are dimensionless. Let $L$, $V$, $\varrho$ be characteristic values of length, speed, and density (in particular, $V$ may be the desired speed introduced in Eq. \eqref{eq.des.speed}) to be used to define the following non-dimensional variables and functions:
\begin{gather*}
	x^\ast=\frac{x}{L}, \quad t^\ast=\frac{V}{L}t, \quad
		\nu^\ast[\mu^\ast_{t^\ast}](x^\ast)=\frac{1}{V}\nu[\mu_{\frac{L}{V}t^\ast}](Lx^\ast), \\[0.2cm]
	f^\ast(s^\ast)=\frac{1}{V}f(Ls^\ast), \quad 
		\rho^\ast(t^\ast,\,x^\ast)=\frac{1}{\varrho}\rho\left(\frac{L}{V}t^\ast,\,Lx^\ast\right), \quad
			P^\ast_j(t^\ast)=\frac{1}{L}P_j\left(\frac{L}{V}t^\ast\right).
\end{gather*}
Notice that, due to the choice of $V$ as characteristic speed, the dimensionless desired velocity $v^\ast_\des$ turns out to be a unit vector.

In more detail, the non-dimensional mass measure $\mu^\ast_{t^\ast}$ is given by
\begin{equation*}
	\begin{array}{rcl}
		d\mu^\ast_{t^\ast}(x^\ast) & = & d\mu_{\frac{L}{V}t^\ast}(Lx^\ast) \\[0.3cm]
		& = & \theta\sum_j\delta_{LP^\ast_j(t^\ast)}(Lx^\ast)+
			(1-\theta)\varrho L^d\rho^\ast(t^\ast,\,x^\ast)\,dx^\ast \\[0.3cm]
		& = & \theta\sum_j\delta_{P^\ast_j(t^\ast)}(x^\ast)+
			(1-\theta)\Lambda\rho^\ast(t^\ast,\,x^\ast)\,dx^\ast \\[0.3cm]
		& = & \theta m^\ast_{t^\ast}(x^\ast)+(1-\theta)\Lambda M^\ast_{t^\ast}(x^\ast)
	\end{array}
\end{equation*}
where we have set $\Lambda:=\varrho L^d$ and we have recognized the dimensionless microscopic and macroscopic masses:
\begin{equation*}
	m^\ast_{t^\ast}=\sum_{j=1}^N\delta_{P^\ast_j(t^\ast)}, \quad
		d M^\ast_{t^\ast}(x^\ast)=\rho^\ast(t^\ast,\,\cdot)\,dx^\ast.
\end{equation*}
We notice that the coefficient $\Lambda$ has unit $[\Lambda]$ = pedestrians, therefore \emph{it is a non-dimensional number fixing the scaling between the microscopic and the macroscopic mass}. It says how many pedestrians are represented, in average, by a unit density $\rho^\ast$ in the infinitesimal reference volume $dx^\ast$.

\begin{remark}
The measure
\begin{equation}
	\mu^\ast_{t^\ast}=\theta m^\ast_{t^\ast}+(1-\theta)\Lambda M^\ast_{t^\ast}
	\label{eq.mu.nondim}
\end{equation}
can be read as a \emph{linear interpolation} between the microscopic and the macroscopic mass via the parameter $\theta$, provided $m^\ast_{t^\ast}(\R^d)$, $M^\ast_{t^\ast}(\R^d)$ are, up to scaling, the \emph{same} mass, \ie, $m^\ast_{t^\ast}(\R^d)=\Lambda M^\ast_{t^\ast}(\R^d)$. As we will see later (cf. Corollary \ref{coroll.prop.mun}), in the multiscale model the microscopic and macroscopic mass are individually conserved in time, hence this can be achieved by setting
\begin{equation}
	\Lambda=\frac{m^\ast_0(\R^d)}{M^\ast_0(\R^d)}=\frac{N}{M^\ast_0(\R^d)}
	\label{eq.lambda}
\end{equation}
as long as $0<N,\,M^\ast_0(\R^d)<+\infty$.
\end{remark}

In the following we will invariably refer to the non-dimensional form of the equations, omitting the asterisks on the non-dimensional variables for brevity.

\section{Discrete-in-time model}
\label{sect.discr.time}
In this section we derive a discrete-in-time counterpart of the multiscale model, that will help us gain some insights into the qualitative properties of the mathematical structures previously outlined. In addition, it will serve as a first step to devise a numerical scheme for the approximate solution of the equations.

Let $\Delta{t}_n>0$ be a possibly adaptive time step and let us introduce a sequence of discrete times $\{t_n\}_{n\geq 0}$ such that $t_0=0$ and $t_{n+1}-t_n=\Delta{t}_n$. Denoting $\mu_n:=\mu_{t_n}$, from Eq. \eqref{eq.cons.mass.weak.solution} with the choice $t_1=t_n$, $t_2=t_{n+1}$ we get
\begin{align*}
	\lint_{\R^d}\phi(x)\,d\mu_{n+1}(x)-\lint_{\R^d}\phi(x)\,d\mu_n(x) &=
		\lint_{t_n}^{t_{n+1}}\lint_{\R^d}v(t,\,x)\cdot\nabla{\phi(x)}\,d\mu_t(x)\,dt \\
	&=\Delta{t}_n\lint_{\R^d}v(t_n,\,x)\cdot\nabla{\phi(x)}\,d\mu_n(x)+o(\Delta{t}_n),
\end{align*}
whence
\begin{equation*}
	\lint_{\R^d}\phi(x)\,d\mu_{n+1}(x)=
		\lint_{\R^d}\left[\phi(x)+\Delta{t}_n\,v(t_n,\,x)\cdot\nabla{\phi(x)}\right]\,d\mu_n(x)+o(\Delta{t}_n).
\end{equation*}
At this point let us explicitly assume that $\mu_n(\R^d)<+\infty$. If $v(t_n,\,\cdot)$ is $\mu_n$-uniformly bounded then $\phi(x)+\Delta{t}_n\,v(t_n,\,x)\cdot\nabla{\phi(x)}=\phi(x+\Delta{t}_n\,v(t_n,\,x))+o(\Delta{t}_n)$, thus
\begin{equation*}
	\lint_{\R^d}\phi(x)\,d\mu_{n+1}(x)=\lint_{\R^d}\phi(x+\Delta{t}_n\,v(t_n,\,x))\,d\mu_n(x)+o(\Delta{t}_n).
\end{equation*}
Defining the \emph{flow map} $\gamma_n(x):=x+v(t_n,\,x)\Delta{t}_n$ and neglecting the term $o(\Delta{t}_n)$, we are finally left with
\begin{equation}
	\lint_{\R^d}\phi(x)\,d\mu_{n+1}(x)=\lint_{\R^d}\phi(\gamma_n(x))\,d\mu_n(x),
	\label{eq.pushfwd.weak}
\end{equation}
which makes sense actually for every bounded and Borel function $\phi$. Choosing $\phi=\chi_E$ for some measurable set $E\in\B(\R^d)$ entails
\begin{equation*}
	\mu_{n+1}(E)=\mu_n(\gamma_n^{-1}(E)), \quad \forall\,E\in\B(\R^d),
\end{equation*}
meaning that $\mu_{n+1}$ is the \emph{push forward} of $\mu_n$ via the flow map $\gamma_n$, also written $\mu_{n+1}=\gamma_n\#\mu_n$. Equation \eqref{eq.pushfwd.weak} provides a discrete-in-time counterpart of Eq. \eqref{eq.cons.mass.weak.solution}. Obviously, it requires to be supplemented by an initial condition $\mu_0$ in order for the sequence $\{\mu_n\}_{n\geq 1}$ to be recursively generated.

Notice that, with the velocity field \eqref{eq.v}, it results $v(t_n,\,x)=v[\mu_n](x)$ with in particular:
\begin{align}
	\nu[\mu_n](x)&=\theta\sum_{\substack{k=1,\,\dots,\,N\\ P_k^n\ne x}}f(\vert P_k^n-x\vert)
		g(\alpha_{xP_k^n})\frac{P_k^n-x}{\vert P_k^n-x\vert} \notag \\
	&\phantom{=}+(1-\theta)\Lambda\lint_{\R^d}f(\vert y-x\vert)g(\alpha_{xy})\frac{y-x}{\vert y-x\vert}
		\rho_n(y)\,dy,
	\label{eq.nu.tdiscr}
\end{align}
where $P_k^n:=P_k(t_n)$ and $\rho_n:=\rho(t_n,\,\cdot)$.

\subsection*{Preserving the multiscale structure of the measure}
Recall that in the multiscale model we assumed that our measure is composed by a microscopic granular and a macroscopic continuous mass\footnote{With respect to the general structure of a measure as provided by Riesz's Theorem, this means that we are in particular excluding the Cantor's part.}. Of course, this is just a formal assumption made to write the model. From the analytical point of view, it need be proved that such a measure can be actually a solution to our equations.

Set $m_n:=m_{t_n}$, $M_n:=M_{t_n}$, so that, owing to Eq. \eqref{eq.mu.nondim}, the measure $\mu_n$ can be given the form
\begin{equation}
	\mu_n=\theta m_n+(1-\theta)\Lambda M_n.
	\label{eq.mun.micromacro}
\end{equation}
The following result clarifies the role played by the flow map $\gamma_n$ in preserving the multiscale structure of $\mu_n$ after one time step.
\begin{theorem}
Let a constant $C_n>0$ exist such that
\begin{equation}
	\Leb^d(\gamma_n^{-1}(E))\leq C_n\Leb^d(E), \quad \forall\,E\in\B(\R^d).
	\label{eq.prop.gamman}
\end{equation}
Given $\mu_n$ as in Eq. \eqref{eq.mun.micromacro}, there exist both a unique atomic measure $m_{n+1}$ and a unique Lebesgue-absolutely continuous measure $M_{n+1}$, with a.e. nonnegative density, such that $\mu_{n+1}=\theta m_{n+1}+(1-\theta)\Lambda M_{n+1}$.
\label{theo.micromacro.preserved}
\end{theorem}
\begin{proof}
1. Using the linearity of the operator $\gamma_n\#\cdot$, the measure $\mu_{n+1}$ is given by
\begin{equation*}
	\mu_{n+1}=\theta(\gamma_n\# m_n)+(1-\theta)\Lambda(\gamma_n\# M_n).
\end{equation*}

2. Let us define
\begin{equation*}
	m_{n+1}:=\gamma_n\# m_n.
\end{equation*}
A direct calculation shows that such a $m_{n+1}$ is in fact an atomic measure. For any measurable set $E\in\B(\R^d)$ we compute:
\begin{align*}
	(\gamma_n\# m_n)(E)=m_n(\gamma_n^{-1}(E))&=\sum_{j=1}^N\delta_{P_j^n}(\gamma_n^{-1}(E)) \\
	&=\operatorname{card}\{\gamma_n(P_j^n)\in E\} \\
	&=\sum_{j=1}^N\delta_{\gamma_n(P_j^n)}(E),
\end{align*}
hence $m_{n+1}$ as defined above is in turn a combination of Dirac masses centered in the new positions $\{P_j^{n+1}\}_{j=1}^N$ given by
\begin{equation}
	P_j^{n+1}:=\gamma_n(P_j^n)=P_j^n+v[\mu_n](P_j^n)\Delta{t}_n.
	\label{eq.new.pos}
\end{equation}

3. Analogously, let us define
\begin{equation*}
	M_{n+1}:=\gamma_n\# M_n.
\end{equation*}
We claim that, under the hypothesis of the theorem, this measure is absolutely continuous w.r.t. $\Leb^d$. To see this, let $E\in\B(\R^d)$ be such that $\Leb^d(E)=0$. Then $\Leb^d(\gamma_n^{-1}(E))\leq C_n\Leb^d(E)=0$, whence, using that $M_n\ll\Leb^d$ by assumption, we get $M_{n+1}(E)=M_n(\gamma_n^{-1}(E))=0$ and the claim follows.

4. To show the non-negativity of the density of $M_{n+1}$ we take the Radon-Nikodym derivative. Then we discover, for a.e. $x$,
\begin{equation*}
	\rho_{n+1}(x)=\lim_{r\to 0^+}\frac{M_{n+1}(B_r(x))}{\Leb^d(B_r(x))}=
		\lim_{r\to 0^+}\frac{1}{\omega_dr^d}\lint_{\gamma_n^{-1}(B_r(x))}\rho_n(y)\,dy,
\end{equation*}
where $\omega_d$ is the volume of the unit ball in $\R^d$. Since $\rho_n$ is a.e. nonnegative by assumption, the same holds for $\rho_{n+1}$ and we are done.

5. Finally, uniqueness of $m_{n+1}$, $M_{n+1}$ is implied by the uniqueness of the Radon-Nikodym decomposition of a measure.
\end{proof}

The proof of Theorem \ref{theo.micromacro.preserved} is constructive, indeed it shows explicitly how to obtain the measure $\mu_{n+1}$ starting from $\mu_n$: one simply pushes forward \emph{separately} the microscopic and the macroscopic mass via the \emph{common} flow map $\gamma_n$.

By referring to the results proved in \cite{piccoli2010tem}, we can state some additional properties of the measure $\mu_n$.
\begin{corollary}
If $\gamma_n$ satisfies \eqref{eq.prop.gamman} for all $n\geq 0$ and the initial measure $\mu_0$ complies with the form \eqref{eq.mun.micromacro} then:
\begin{enumerate}
\item there exists a unique sequence of atomic measures $\{m_n\}_{n\geq 1}$ and a unique sequence of positive Lebesgue-absolutely continuous measures $\{M_n\}_{n\geq 1}$ such that $\mu_n$ has the form \eqref{eq.mun.micromacro} for all $n\geq 1$;
\item the measure $\mu_n$ satisfies the following conservation law:
\begin{equation}
	\mu_{n+1}(E)-\mu_n(E)=-\left[\mu_n(E\setminus\gamma_n^{-1}(E))-\mu_n(\gamma_n^{-1}(E)\setminus E)\right]
	\label{eq.mun.cons.law}
\end{equation}
for all $E\in\B(\R^d)$. In particular, both $m_n$ and $M_n$ satisfy this law separately at each time step;
\item if $\rho_0\in L^\infty_{\text{loc}}(\R^d)$ then $\rho_n\in L^\infty_{\text{loc}}(\R^d)$ for all $n\geq 1$, with moreover
\begin{equation*}
	\esssup_{x\in E}{\vert\rho_n(x)\vert}\leq\prod_{j=0}^{n-1}C_j\esssup_{x\in E}{\vert\rho_0(x)\vert}
\end{equation*}
for all compact set $E\subset\R^d$, where the $C_j$'s are those appearing in the statement of Theorem \ref{theo.micromacro.preserved}.
\end{enumerate}
\label{coroll.prop.mun}
\end{corollary}
\begin{proof}
1. Existence and uniqueness of microscopic and macroscopic masses have been proved in Theorem \ref{theo.micromacro.preserved} for one time step, hence they follow for all times $n\geq 1$ by induction.

2. In view of the $\sigma$-additivity of the measure we have, for all $E\in\B(\R^d)$,
\begin{equation*}
	\mu_{n+1}(E)=\mu_n(\gamma_n^{-1}(E)\cap E)+\mu_n(\gamma_n^{-1}(E)\setminus E).
\end{equation*}
Subtracting $\mu_n(E)$ at both sides and collecting conveniently gives
\begin{equation*}
	\mu_{n+1}(E)-\mu_n(E)=-\left[\mu_n(E)-\mu_n(\gamma_n^{-1}(E)\cap E)\right]+\mu_n(\gamma_n^{-1}(E)\setminus E)
\end{equation*}
whence, observing that $E=(\gamma_n^{-1}(E)\cap E)\cup(E\setminus\gamma_n^{-1}(E))$ with disjoint union, the thesis follows. Since we know from Theorem \ref{theo.micromacro.preserved} that $m_{n+1}$, $M_{n+1}$ are in turn generated by push forward with $\gamma_n$, this reasoning can be repeated to find that each of them fulfills the very same conservation law.

3. In \cite{piccoli2010tem} it is proved that $\rho_n\in L^\infty_{\text{loc}}(\R^d)$ implies $\rho_{n+1}\in L^\infty_{\text{loc}}(\R^d)$ as well, with $\esssup_{x\in E}{\vert\rho_{n+1}(x)\vert}\leq C_n\esssup_{x\in E}{\vert\rho_n(x)\vert}$. Thus proceeding by induction from $n=0$ we get the result.
\end{proof}

Some comments on the results of this section are in order.
\begin{enumerate}
\item [(i)] The main assumption of both Theorem \ref{theo.micromacro.preserved} and Corollary \ref{coroll.prop.mun} is that the flow map $\gamma_n$ satisfies Eq. \eqref{eq.prop.gamman}. In general it may be hard to check the validity of this property directly but in \cite{piccoli2010tem} it is proved that a sufficient condition for it to hold true is that the velocity $v[\mu_n]$ be Lipschitz continuous and that the time step be chosen so that $\Delta{t}_n\operatorname{Lip}(v[\mu_n])<1$, $n=0,\,1,\,2,\,\dots$.
\item [(ii)] Equation \eqref{eq.mun.cons.law} states that the variation of the mass of a set $E$ in one time step is due to the net mass inflow or outflow across $\partial E$. Indeed, $E\setminus\gamma_n^{-1}(E)$ is the subset of $E$ which is not mapped into $E$ by $\gamma_n$ (outgoing flux), and $\gamma_n^{-1}(E)\setminus E$ is the subset of $\R^d\setminus E$ which is mapped into $E$ by $\gamma_n$ (ingoing flux).
\item [(iii)] Choosing $E=\R^d$ in Eq. \eqref{eq.mun.cons.law} yields $\mu_{n+1}(\R^d)=\mu_n(\R^d)$, \ie, the conservation of the total mass at each time step. If $\mu_0(\R^d)<+\infty$ then the mass is finite for all $n$, therefore, up to normalization, the $\mu_n$'s may be regarded as probability measures. Furthermore, since Eq. \eqref{eq.mun.cons.law} applies separately also to $m_n$ and $M_n$, both the total microscopic and the total macroscopic masses are conserved in time.
\end{enumerate}

\section{Numerical approximation of the equations}
\label{sect.num.approx}
As anticipated at the beginning of Section \ref{sect.discr.time}, the discrete-in-time model provides a first discretization of the equations, which would be sufficient for tracking the microscopic mass (cf. Eq. \eqref{eq.new.pos}). However, the macroscopic mass requires a further discretization in space in order to come to a full approximation of the density $\rho$. Notice that, in principle, one may refer to Eq. \eqref{eq.continuity.equation} and rely on the wide literature on numerical methods for nonlinear hyperbolic conservation laws \cite{MR1925043}. Nevertheless, aside from the intrinsic complication due to the multidimensional nature of the equations, this strategy poses several nontrivial technical difficulties. For example, it demands a correct definition of the convection velocity (\ie, formally the derivative of the flux $\rho v$ with respect to the density) in presence of nonlocal multiscale fluxes, as well as a consistent formulation of entropy-like criteria for picking up physically significant solutions. All these issues are instead bypassed if one maintains the measure-theoretic formalism.

For the discretization in space of the density $\rho_n$ we partition the domain in pairwise disjoint $d$-dimensional cells $E_i\in\B(\R^d)$, where $i\in\Z^d$ is an integer multi-index, sharing a characteristic size $h>0$ such that $\Leb^d(E_i)\to 0$ for all $i$ when $h\to 0^+$ (for instance, $h\sim\operatorname{diam}E_i$). Every cell is further identified by one of its points $x_i$, e.g., its center in case of regular cells.

We approximate $\rho_n$ by a piecewise constant function $\tilde{\rho}_n$ on the numerical grid:
\begin{equation*}
	\tilde{\rho}_n(x)\equiv\rho_i^n, \quad \forall\,x\in E_i,
\end{equation*}
where $\rho_i^n\geq 0$ is the value that $\tilde{\rho}_n$ takes in the cell $E_i$. Consequently, the measure $M_n$ is approximated by the piecewise constant measure $d\tilde{M}_n=\tilde{\rho}_n\,d\Leb^d$, which entails the approximation $\tilde{\mu}_n=\theta m_n+(1-\theta)\Lambda\tilde{M}_n$ for $\mu_n$.

Analogously, we approximate the velocity $v[\mu_n]$ by a piecewise constant field
\begin{equation*}
	\tilde{v}[\tilde{\mu}_n](x)\equiv v_i^n, \quad \forall\,x\in E_i,
\end{equation*}
where the values $v_i^n\in\R^d$ are computed as $v_i^n=v[\tilde{\mu}_n](x_i)$. The discretization of the velocity gives rise to the following discrete flow map:
\begin{equation*}
	\tilde{\gamma}_n(x)=x+\tilde{v}[\tilde{\mu}_n](x)\Delta{t}_n,
\end{equation*}
which turns out to be a piecewise translation because $\tilde{v}[\tilde{\mu}_n]$ is constant in each cell.

Finally, we look for a piecewise constant approximation $\tilde{M}_{n+1}$ of $M_{n+1}$ by imposing the push forward of $\tilde{M}_n$ via the flow map $\tilde{\gamma}_n$:
\begin{equation*}
	\tilde{M}_{n+1}(E)=\tilde{M}_n(\tilde{\gamma}_n^{-1}(E)), \quad \forall\,E\in\B(\R^d).
\end{equation*}
In particular, choosing $E=E_i$ yields
\begin{equation}
	\rho_i^{n+1}=\frac{1}{\Leb^d(E_i)}\sum_{k\in\Z^d}\rho_k^n\Leb^d(E_i\cap\tilde{\gamma}_n(E_k)),
		\quad \forall\,i\in\Z^d,
	\label{eq.num.scheme.macro}
\end{equation}
which provides a time-explicit scheme to compute the coefficients of the density $\tilde{\rho}_{n+1}$ from those of $\tilde{\rho}_n$. Notice in particular that $\tilde{\gamma}_n(E_k)$ is simply the set $E_k+v_k^n\Delta{t}_n$.

Notice that this scheme is positivity-preserving, in the sense that $\tilde{\rho}_n\geq 0$ implies $\tilde{\rho}_{n+1}\geq 0$ as well, hence, by induction, $\tilde{\rho}_0\geq 0$ implies $\tilde{\rho}_n\geq 0$ for all $n>0$. Such a basic property is not as straightforward in usual numerical schemes for hyperbolic conservation laws. Indeed, unless suitable corrections are implemented, the latter may develop oscillations leading to locally negative approximate solutions even when the exact solution is not expected to be so.

Furthermore, considering that $\tilde{\gamma}_n$ is a translation in each grid cell and using the invariance of the Lebesgue measure under rigid transformations, we deduce
\begin{equation*}
	\lint_{\R^d}\tilde{\rho}_{n+1}(x)\,dx=
		\sum_{k\in\Z^d}\rho_k^n\sum_{i\in\Z^d}\Leb^d(\tilde{\gamma}_n^{-1}(E_i)\cap E_k)
		=\sum_{k\in\Z^d}\rho_k^n\Leb^d(E_k)=\lint_{\R^d}\tilde{\rho}_n(x)\,dx,
\end{equation*}
thus the approximate macroscopic mass $\tilde{M}_n$ is conserved in time.

The quality of the spatial discretization described above with respect to the refinement of the grid, in the case of regular flow maps, is provided by the following result.
\begin{remark}
At this point we assume explicitly that the domain of the problem is a \emph{bounded} set $\Omega\subset\R^d$, which for all fixed $h>0$ is partitioned with a \emph{finite} number of grid cells (however tending to infinity when $h\to 0^+$). The multi-index $i$ of the grid cells runs in a finite subset $\I\subset\Z^d$.
\end{remark}
\begin{theorem}
Assume that $\gamma_n$ is a diffeomorphism and let $h,\,\Delta{t}_n$ be sufficiently small and satisfying
\begin{equation}
	\max_{i\in\I}\frac{\Delta{t}_n}{h}\vert v_i^n\vert\leq 1.
	\label{eq.CFL-like}
\end{equation}
Then:
\begin{enumerate}
\item[(i)] There exists a constant $C_n>0$, independent of $h$, such that
\begin{equation*}
	\sum_{i\in\I}\vert M_{n+1}(E_i)-\tilde{M}_{n+1}(E_i)\vert\leq C_n\left(
		\lint_{\Omega}\vert\rho_n(x)-\tilde{\rho}_n(x)\vert\,dx+h\right).
\end{equation*}
\item[(ii)] If $v[\mu_n](x)$ is uniformly bounded, there exists a constant $C'_n>0$, independent of $h$, such that
\begin{equation*}
	\max_{i\in\I}\vert M_n(E_i)-\tilde{M}_n(E_i)\vert\leq
		\max_{i\in\I}\vert M_0(E_i)-\tilde{M}_0(E_i)\vert+C'_nh^d.
\end{equation*}
\end{enumerate}
\label{theo.spat.discr}
\end{theorem}
\begin{proof}
See \cite{piccoli2010tem}.
\end{proof}

In order to gain some control over the error introduced by the spatial discretization, Theorem \ref{theo.spat.discr} requires the CFL-like condition \eqref{eq.CFL-like} to be satisfied at each time step, similarly to numerical schemes for hyperbolic conservation laws. However there is a remarkable difference from their CFL condition, namely that Eq. \eqref{eq.CFL-like} involves directly the flux velocity and \emph{not} the convection velocity.

\subsection*{The algorithm}
Here we detail the numerical algorithm stemming from the above scheme, that we use for simulations.

The algorithm combines a microscopic and a macroscopic part. The former handles the evolution of pedestrian positions, updating a vector which stores the values $P_j^n\in\R^d$. The latter manages instead the evolution of the density, and at every time step it updates the values $\rho_i^n$ at the grid cells. The two models evolve by means of the same velocity field $\tilde{v}[\tilde{\mu}_n]$, thus guaranteeing coherence of the final solution. This conceptual scheme is motivated by Theorem \ref{theo.micromacro.preserved}. The velocity field must be defined at pedestrian positions $\{P_j^n\}_{j=1}^N$ for the microscopic part and at the grid cells $\{E_i\}_{i\in\I}$ for the macroscopic part.

Let us introduce the following superscripts:
\begin{itemize}
\item \emph{micro}: quantities defined at pedestrian positions;
\item \emph{macro}: quantities defined at grid cells;
\item \emph{micro-for-micro}: microscopic quantities computed at pedestrian positions;
\item \emph{micro-for-macro}: microscopic quantities computed at grid cells;
\item \emph{macro-for-micro}: macroscopic quantities computed at pedestrian positions;
\item \emph{macro-for-macro}: macroscopic quantities computed at grid cells.
\end{itemize}

The algorithm consists of the following steps.
\begin{enumerate}
\item \emph{Initialization}. We fix the number $N$ of microscopic pedestrians that we want to model, we define their positions, and we compute the coefficients $\rho_i^0$ of the initial density according to a local average of the microscopic mass, taking the scaling \eqref{eq.lambda} into account. More precisely, we set:
\begin{equation*}
	\rho_i^0=\frac{m_0(B_\xi(x_i))}{\Lambda\Leb^d(B_\xi(x_i))}, \quad i\in\Z^d,
\end{equation*}
where $m_0$ is the microscopic mass at the initial time and $B_\xi(x_i)$ is the ball centered in the center of the grid cell $E_i$ with radius $\xi>0$. The latter is tuned depending on the positions of the microscopic pedestrians, in such a way that the relation $\Lambda=m_0(\R^d)/\tilde{M}_0(\R^d)$ be satisfactorily fulfilled in the numerical sense\footnote{Notice that if one replaces $B_\xi(x_i)$ with the cell $E_i$ then the measures $m_0$, $\tilde{M}_0$ satisfy the scaling \eqref{eq.lambda} exactly. However, averaging on a neighborhood a bit larger than a single grid cell is essential in order to have a macroscopic density really distributed in space rather than clustered in grid cells.} ($\tilde{M}_0$ being the approximate macroscopic mass at the initial time).
\item \emph{Microscopic part}. At time $t=t_n$ we compute the sum at the right-hand side of Eq. \eqref{eq.nu.tdiscr} for $x=P_j^n$ obtaining
\begin{equation*}
	\tilde{\nu}^\text{micro-for-micro}:=\tilde{\nu}[m_n](P_j^n).
\end{equation*}
The same computation performed for $x=x_i$ gives instead
\begin{equation*}
	\tilde{\nu}^\text{micro-for-macro}:=\tilde{\nu}[m_n](x_i),
\end{equation*}
cf. Eq. \eqref{eq.nu.micro-for-macro}, which will be shared with the macroscopic part of the code.
\item \emph{Macroscopic part}. At the same time instant $t=t_n$ we numerically evaluate the integral at the right-hand side of Eq. \eqref{eq.nu.tdiscr} for $x=x_i$, using the approximate density $\tilde{\rho}_n$ in place of $\rho_n$. This way we obtain
\begin{equation*}
	\tilde{\nu}^\text{macro-for-macro}:=\tilde{\nu}[\tilde{M}_n](x_i).
\end{equation*}
Next we compute the same integral for $x=P_j^n$, which yields
\begin{equation*}
	\tilde{\nu}^\text{macro-for-micro}:=\tilde{\nu}[\tilde{M}_n](P_j^n),
\end{equation*}
cf. Eq. \eqref{eq.nu.macro-for-micro}. In particular, the integrals involved in $\tilde{\nu}^\text{macro-for-macro}$ and $\tilde{\nu}^\text{macro-for-micro}$ are numerically evaluated via a first order quadrature formula. This component of the velocity field will be shared with the microscopic part of the code.
\item \emph{Desired velocity}. If the velocity field $v_\des$ is given analytically, the computation of $v_\des^\text{micro}:=v_\des(P_j^n)$ and of $v_\des^\text{macro}:=v_\des(x_i)$ is immediate. If instead $v_\des$ is defined on the numerical grid only, for instance because it comes from the numerical solution of other equations \cite{piccoli2009pfb}, then $v_\des^\text{micro}$ is computed by interpolation. Since we are assuming that all macroscopic quantities are piecewise constant, we coherently choose a zeroth order interpolation.

\item \textit{Overall velocity}. We assemble the previous pieces as 
\begin{align*}
	\tilde{v}^\text{micro}:&=\tilde{v}[\tilde{\mu}_n](P_j^n)\\
	&=v_\des^\text{micro}+
		\theta\tilde{\nu}^\text{micro-for-micro}+(1-\theta)\Lambda\tilde{\nu}^\text{macro-for-micro},
\end{align*}
and analogously
\begin{align*}
	\tilde{v}^\text{macro}:&=\tilde{v}[\tilde{\mu}_n](x_i)\\
	&=v_\des^\text{macro}+
		\theta\tilde{\nu}^\text{micro-for-macro}+(1-\theta)\Lambda\tilde{\nu}^\text{macro-for-macro}.
\end{align*}
\item \textit{Computation of $\Delta t$}. We compute the largest time step $\Delta t$ allowed by condition \eqref{eq.CFL-like} for the macroscopic velocity field $\tilde{v}^\text{macro}$.

\item \textit{Advancing in time}. We update pedestrian positions and density according to Eqs. \eqref{eq.new.pos}, \eqref{eq.num.scheme.macro} by means of $\tilde{v}^\text{micro}$ and $\tilde{v}^\text{macro}$, respectively.
\end{enumerate}

\begin{remark}
No matter what the value of $\theta$ is, the two approaches always coexist. If $\theta=0$ the macroscopic scale is leading, and the microscopic pedestrians are simply driven by the macroscopic velocity field. This is the classical way to see flowing (Lagrangian) particles in a fluid, whose motion was previously computed. Conversely, if $\theta=1$ the microscopic scale is leading, and the evolution of the macroscopic density is reliable only if the number of microscopic pedestrians is sufficiently large.
\end{remark}

\section{Numerical tests}
\label{sect.num.res}
In this section we present the results of numerical simulations performed with the model and the algorithm described above. As natural for pedestrian flows, we deal with two-dimensional ($d=2$) bounded domains, say $\Omega\subset\R^2$, confining the attention to the restriction measure $\mu_t\llcorner\Omega$. This means that the mass possibly flowing out of the domain is considered as lost, \ie, it no longer affects the computation.

Sometimes we will deal with domains with obstacles understood as internal holes. They require a careful handling of the velocity at their boundaries so as to prevent it from pointing inward (which would imply unrealistic outflow of mass). In order to have the mass bypass the obstacles, the velocity \eqref{eq.v} is projected onto a \emph{space of admissible velocities} $\V$, which can be defined in several ways depending on the pedestrian behavior one wants to model. Our choice for the next examples is
\begin{equation*}
	\V=\{v\in\R^d\,:\,v\cdot\mathbf{n}\geq 0\ \text{at every obstacle boundary}\},
\end{equation*}
where $\mathbf{n}$ is the outward normal unit vector at the obstacle boundaries. This corresponds to setting to zero the normal component of the velocity \eqref{eq.v} in case it points into an obstacle. A different possibility is to set to zero both the normal and the tangential component if the first one points into an obstacle. In the former case pedestrians can slide along the obstacle walls following the tangential velocity, whereas in the latter case they remain still against the obstacles until no longer pushed by flowing neighbors. This choice may model, for instance, a more relaxed condition in which walkers are not in a hurry to reach their destination.

Concerning the parameters, we assume $\bar\alpha=\pi/2$ (frontal interaction) in all tests, which is suitable for the most common situations encountered in pedestrian flow. We also assume no attraction between group mates but in the last test (Test 4) in which we model the dynamics of a group of people following a leader. Table \ref{tab.parametri} summarizes the values of all other parameters used in the numerical tests.

\begin{table}[!t]
\caption{Summary of the parameters used in the numerical tests}
\label{tab.parametri}
\begin{center}
\begin{tabular}{|c|c|c|c|c|c|c|c|c|}
\hline
Test N. & $\theta$ & $N$ & $\Lambda$ & $F_r$ & $F_a$ & $R_r$ & $R_a$ \\
\hline
\hline
$1$ & $[0,\,1]$ & $100$ & $10$ & $0.1$ & $0$ & $0.5$ & N/A \\
\hline
$2$ & $[0,\,1]$ & $10,\,100$ & $10$, $100$ & $0.1$ & $0$ & $0.25$ & N/A \\
\hline
$3$ & $0$, $0.3$, $1$ & $30$ & $30$ & $0.1$ & $0$ & see text & N/A \\
\hline
$4$ & $0.3$ & $25+1$ & $80$ & $0.05$ & $0.4$ & $1.5$ & $1.5$ \\
\hline
\end{tabular}
\end{center}
\end{table}

\subsection*{Test 1: Dynamics of the interactions}
In this first test we study the effect of the multiscale coupling on the rearrangement of a crowd subject only to internal repulsion. The goal is to show that it is possible to obtain a perfect correspondence between the microscopic and the macroscopic dynamics in some simple cases, which originally motivated and justified the possibility of a coupled multiscale approach \cite{cristiani2010mso}.

To this purpose we switch the desired velocity off, so that the velocity $v[\mu_t]$ coincides with the interaction velocity $\nu[\mu_t]$. (Actually, in order to compute the angle $\alpha_{xy}$ in Eq. \eqref{eq.nu.tdiscr} we conventionally assume $v_\des$ to be constantly directed along the horizontal axis, so as to define what is ahead and what is behind). Pedestrians are initially arranged in a square-shaped equally-spaced formation, see Fig. \ref{fig.test1.snapshots}a. Due to the frontal repulsion, we expect the frontal part of the group to stand still and the rear part to stand back from the group mates ahead. We compare the expansion dynamics of the group as predicted by the macroscopic ($\theta=0$), the microscopic ($\theta=1$), and the multiscale ($\theta=0.3$) model. The simulation runs until the final time  $T=1$ is reached. Notice that the configuration assumed at that time is not an equilibrium of the system. Results are shown in Figs. \ref{fig.test1.snapshots}b-d.

\begin{figure}[!t]
\begin{center}
\begin{minipage}[c]{0.3\textwidth}
	\centering
	\includegraphics[width=\textwidth]{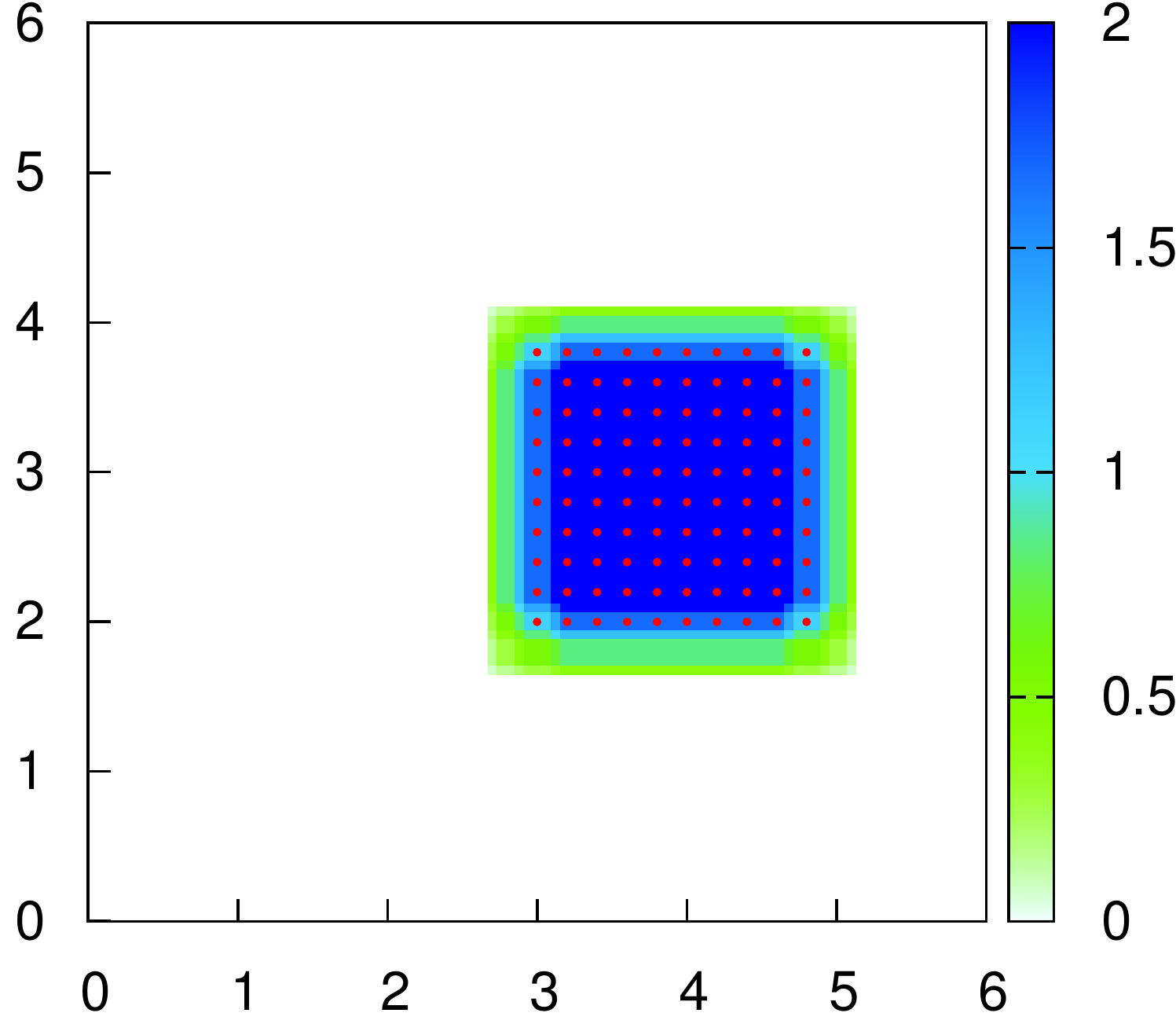} \\
	(a)
\end{minipage}
\qquad
\begin{minipage}[c]{0.3\textwidth}
	\centering
	\includegraphics[width=\textwidth]{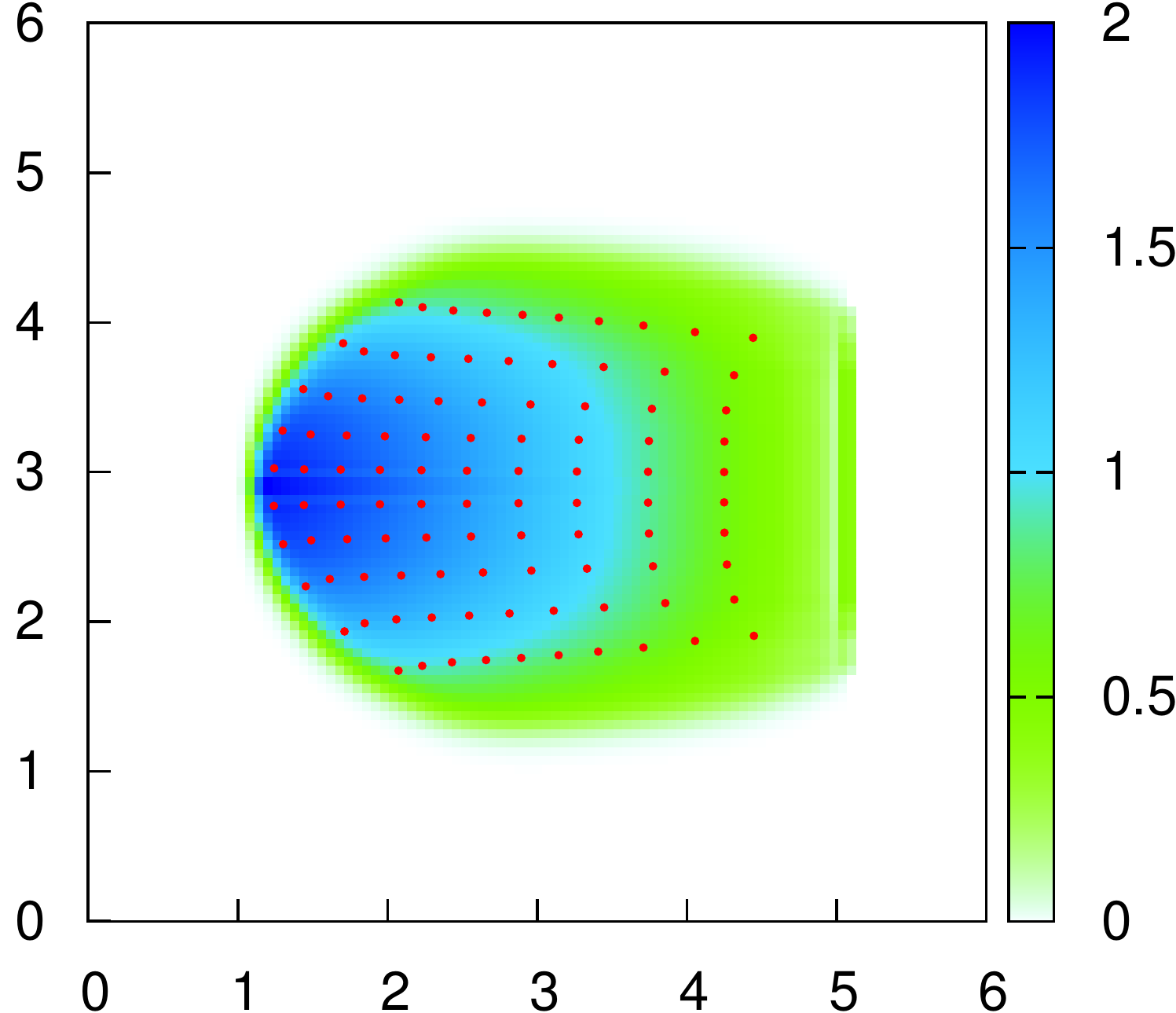} \\
	(b)
\end{minipage} \\[0.3cm]
\begin{minipage}[c]{0.3\textwidth}
	\centering
	\includegraphics[width=\textwidth]{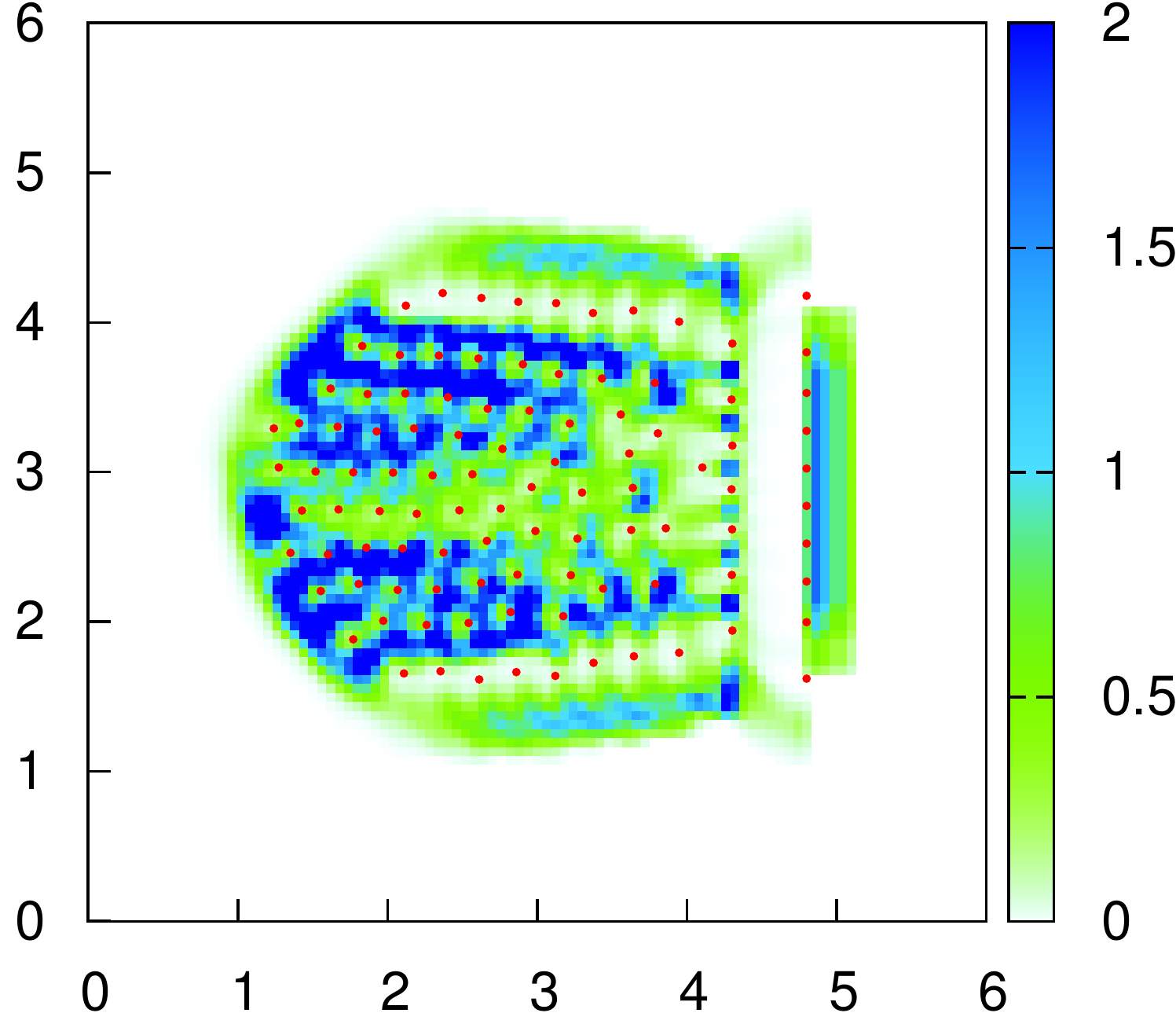} \\
	(c)
\end{minipage}
\qquad
\begin{minipage}[c]{0.3\textwidth}
	\centering
	\includegraphics[width=\textwidth]{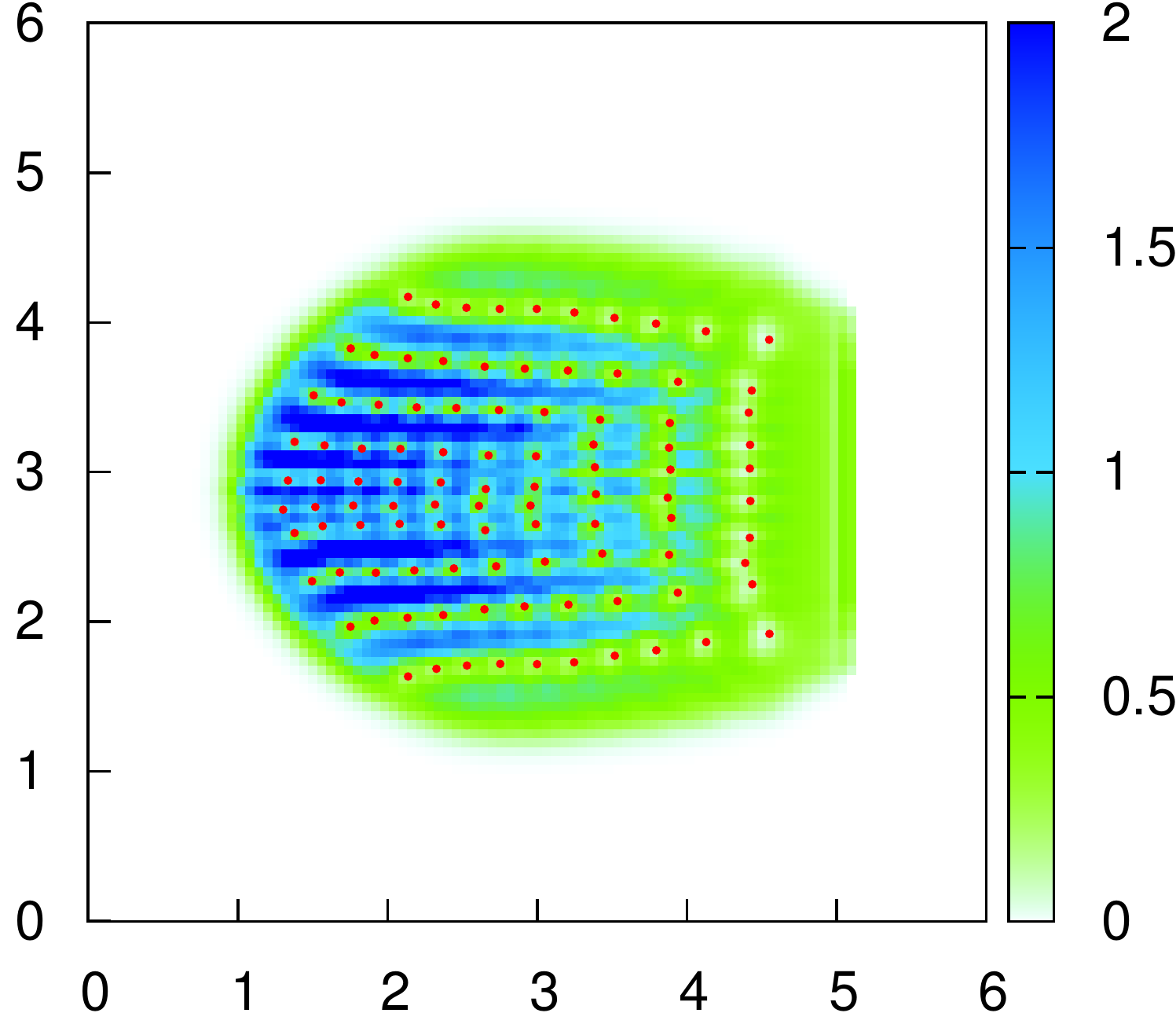} \\
	(d)
\end{minipage}
\caption{Test 1. (a) Initial condition. Crowd distribution at time $T=1$ with (b) the purely macroscopic model, (c) the purely microscopic model, and (d) the multiscale model with $\theta=0.3$}
\label{fig.test1.snapshots}
\end{center}
\end{figure}

The main features of the dynamics outlined above are caught at all scales. In particular, the effect of the only frontal repulsion is visible at the head of the group, where pedestrians stay aligned on a vertical line as they are initially because there is none in front of them. This clearly shows up looking both at the density distribution at the macroscopic scale (Fig. \ref{fig.test1.snapshots}b) and at the individual pedestrians at the microscopic scale (Fig. \ref{fig.test1.snapshots}c).

Of course, this does not mean that either scale has no influence at all on the other. For instance, as an interesting effect of the microscopic scale driving the macroscopic dynamics, we notice some kind of ``density holes'' near every microscopic pedestrian in the limit of the purely microscopic model (Fig. \ref{fig.test1.snapshots}c). They are actually small areas of very low density, caused by that microscopic repulsion has a great impact at the macroscopic scale. Recall indeed that the microscopic granularity is seen as a singularity in the average crowd distribution, and that for $\theta=1$ the evolution of the macroscopic density is fully ruled by the microscopic scale. With the multiscale model (Fig. \ref{fig.test1.snapshots}d with $\theta=0.3$) the hole effect is instead limited, and a good compromise between the two scales is reached. Furthermore, in Fig. \ref{fig.test1.snapshots}d (multiscale) pedestrians are less scattered than in Fig. \ref{fig.test1.snapshots}c (microscopic), meaning that the contribution of the macroscopic scale on the overall dynamics has, in a sense, a homogenizing effect. Conversely, in Fig. \ref{fig.test1.snapshots}d the macroscopic density is more broken than in Fig. \ref{fig.test1.snapshots}b (macroscopic), thus the microscopic scale destroys the macroscopic smoothness and introduces a non-negligible granular effect in the overall dynamics.

To investigate more in depth the intercorrelation between the scales we consider now how the moments of inertia of the mass distribution depend on the coupling parameter $\theta$, fixing all other parameters as indicated in Table \ref{tab.parametri}. Indeed from classical mechanics it is known that moments of inertia provide quantitative information on the shape of the group.

Let $x_G$ be the center of mass of the crowd at the final time $T$:
\begin{equation*}
	x_G=\frac{1}{\mu_T(\Omega)}\lint_\Omega x\,d\mu_T(x),
\end{equation*}
then we consider the following three moments of inertia around $x_G$:
\begin{equation*}
	I_1^\mu=\lint_\Omega\vert(x-x_G)\cdot\mathbf{i}\vert^2\,d\mu_T(x), \quad
	I_2^\mu=\lint_\Omega\vert(x-x_G)\cdot\mathbf{j}\vert^2\,d\mu_T(x), \quad
	I_G^\mu=I_1^\mu+I_2^\mu,
\end{equation*}
$\mathbf{i}$, $\mathbf{j}$ being the unit vectors in the direction of the horizontal and vertical axis, respectively. $I_1^\mu$ and $I_2^\mu$ refer to stretching or shrinking of the group in the horizontal and vertical direction, respectively, whereas $I_G^\mu$ accounts for the global distribution of the crowd around its center of mass. By replacing $\mu_T$ in the above formulas with the measure $m_T$ ($M_T$, resp.) it is possible to study the analogous moments of inertia of the sole microscopic (macroscopic, resp.) mass, that we denote by $I_{1,\,2,\,G}^m$ ($I_{1,\,2,\,G}^M$, resp.).

\begin{figure}[!t]
\begin{center}
\begin{minipage}[c]{0.3\textwidth}
	\centering
	\includegraphics[width=\textwidth]{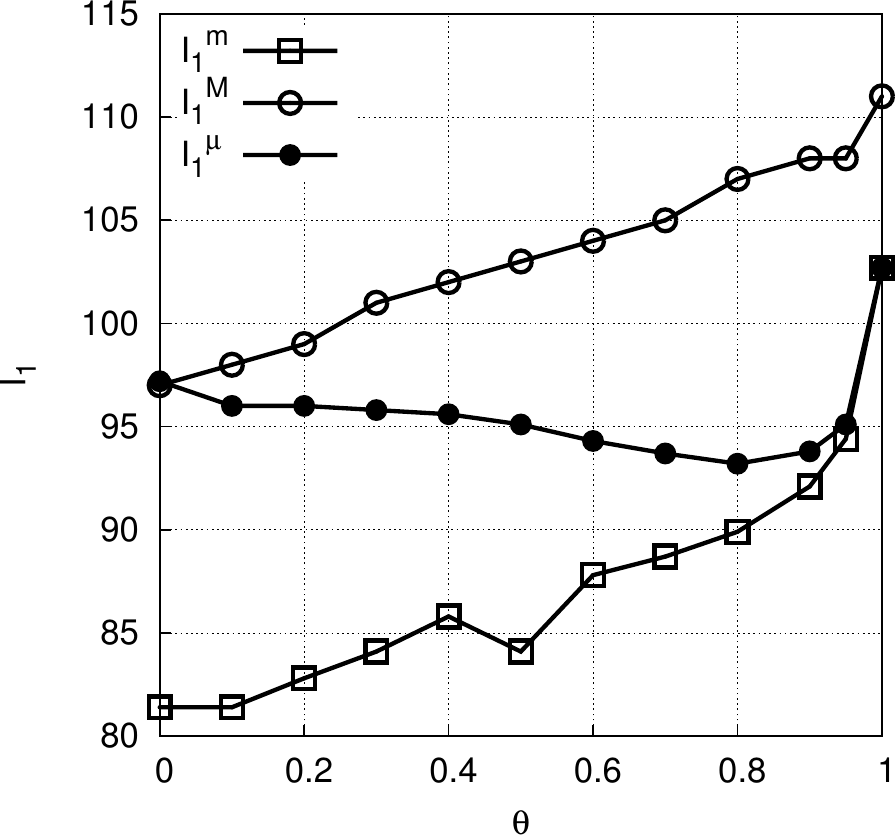} \\
	(a)
\end{minipage}
\hspace{0.3cm}
\begin{minipage}[c]{0.3\textwidth}
	\centering
	\includegraphics[width=\textwidth]{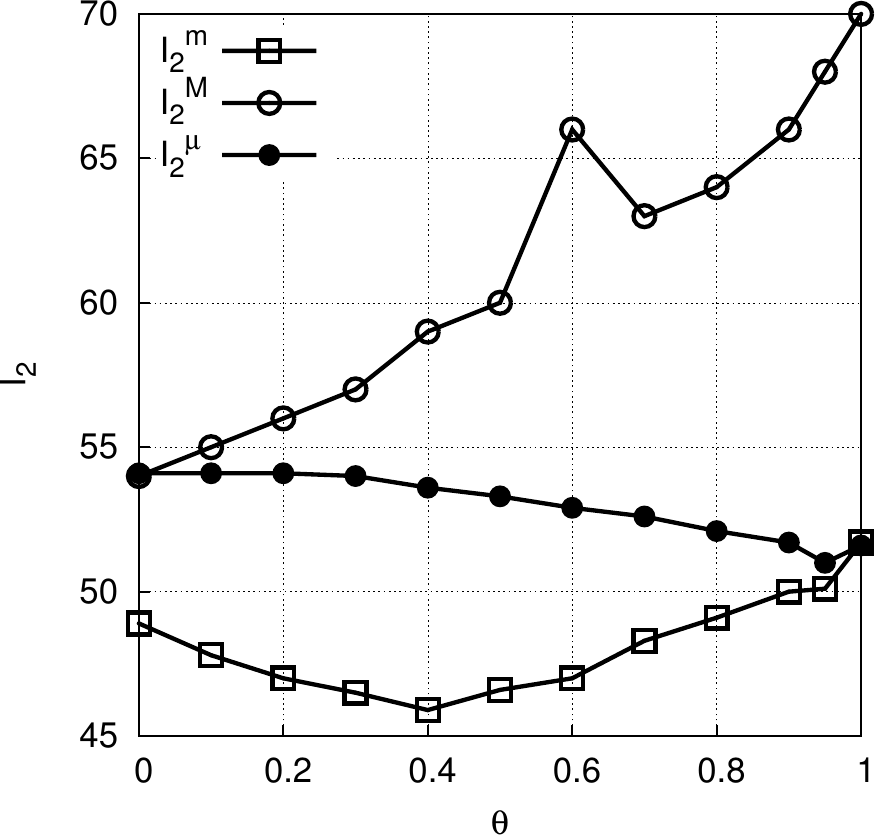} \\
	(b)
\end{minipage}
\hspace{0.3cm}
\begin{minipage}[c]{0.3\textwidth}
	\centering
	\includegraphics[width=\textwidth]{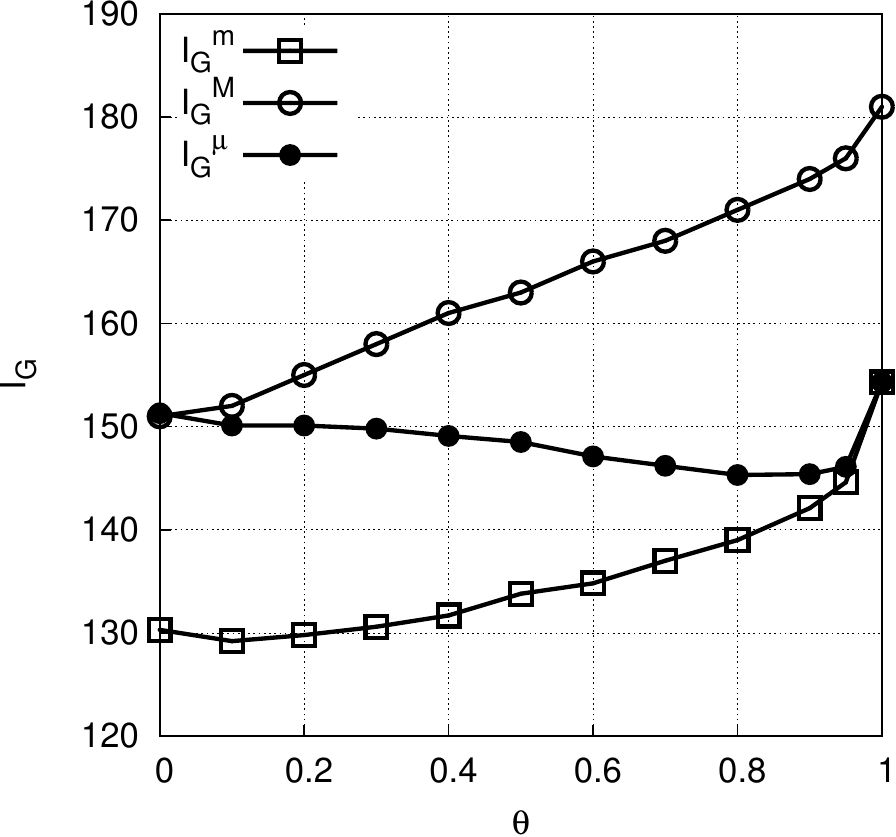} \\
	(c)
\end{minipage}
\caption{Test 1. Moments of inertia of the crowd distribution as functions of $\theta$: (a) $I_1^{m,\,M,\,\mu}$, (b) $I_2^{m,\,M,\,\mu}$, (c) $I_G^{m,\,M,\,\mu}$}
\label{fig.test1.indicators}
\end{center}
\end{figure}
\begin{figure}[!t]
\begin{center}
\begin{minipage}[c]{0.3\textwidth}
	\centering
	\includegraphics[width=\textwidth]{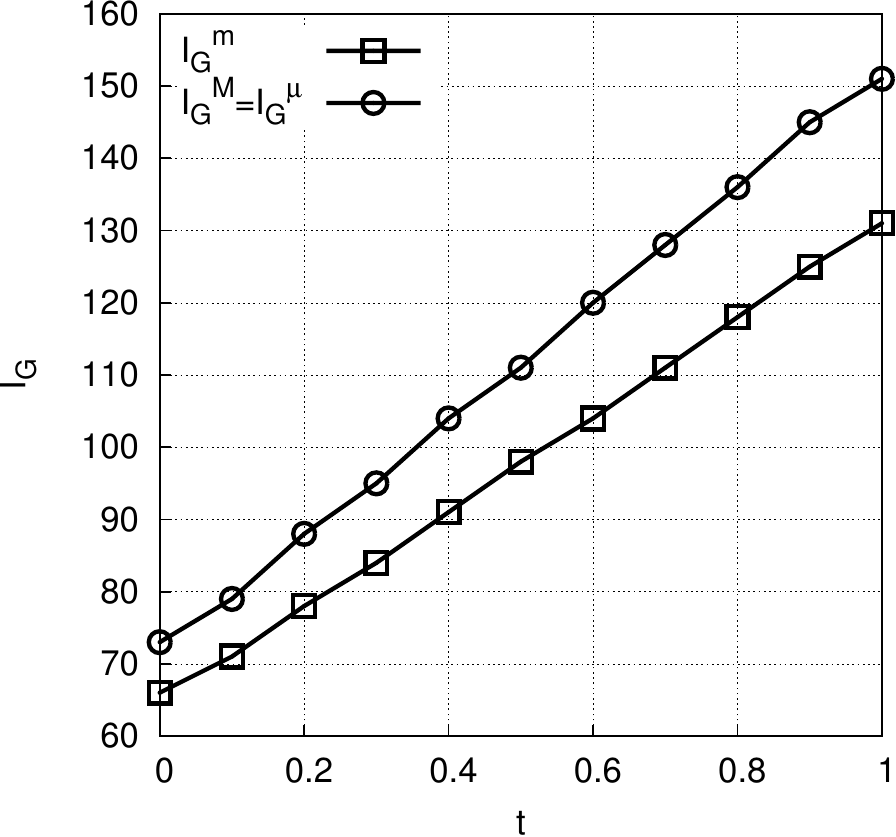} \\
	(a)
\end{minipage}
\hspace{0.3cm}
\begin{minipage}[c]{0.3\textwidth}
	\centering
	\includegraphics[width=\textwidth]{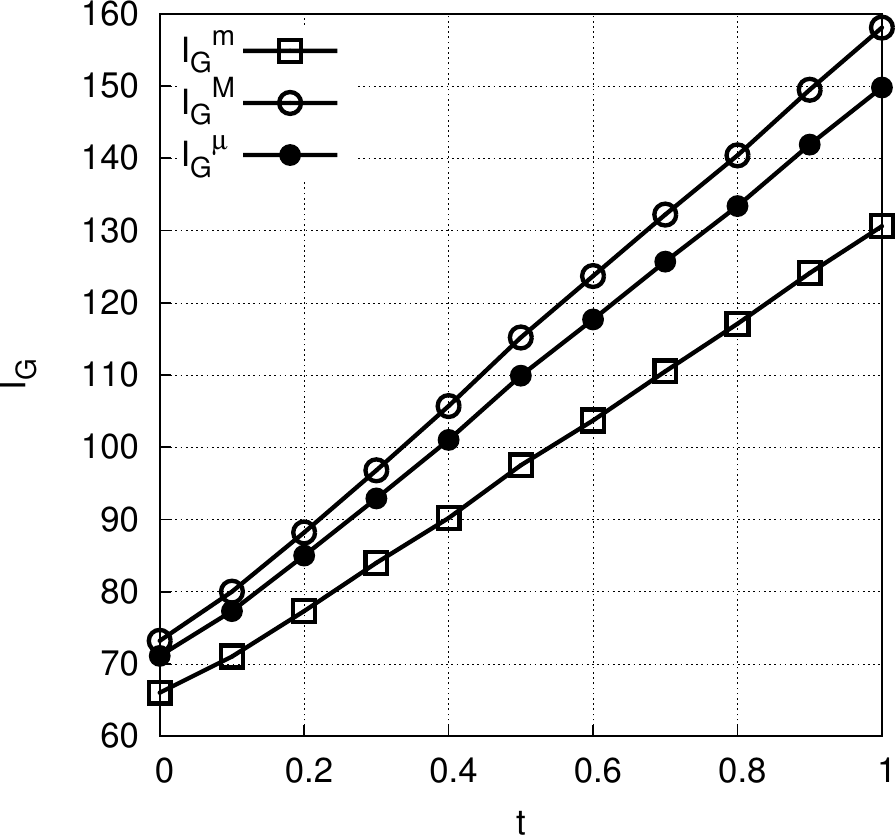} \\
	(b)
\end{minipage}
\hspace{0.3cm}
\begin{minipage}[c]{0.3\textwidth}
	\centering
	\includegraphics[width=\textwidth]{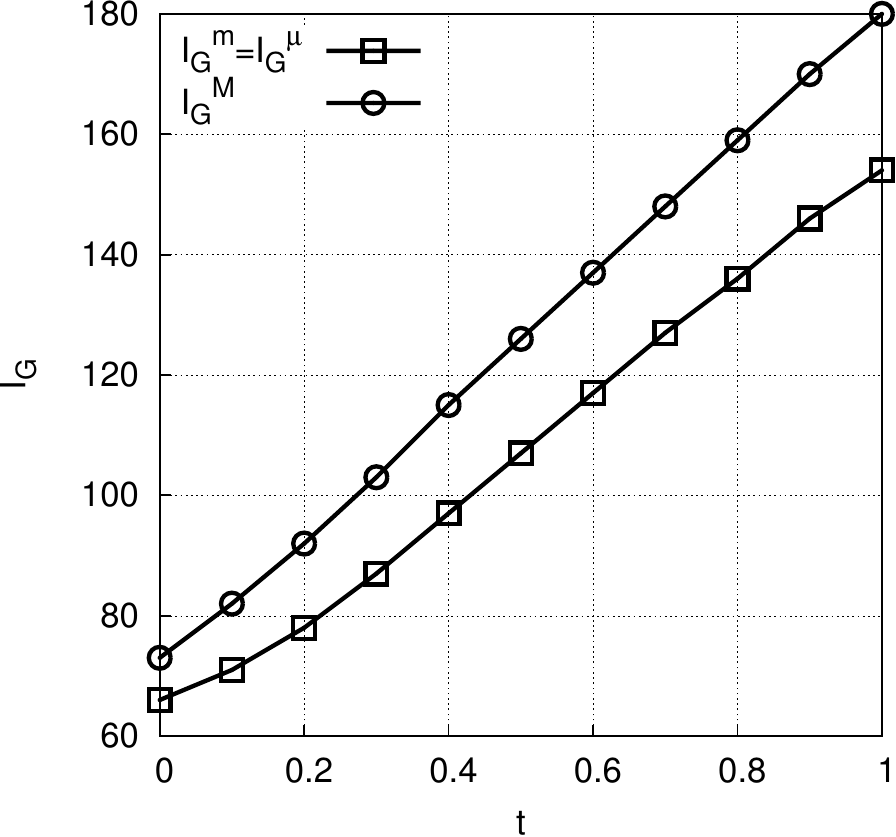} \\
	(c)
\end{minipage}
\caption{Test 1. Moments of inertia $I_G^m$, $I_G^M$, $I_G^{\mu}$ as functions of time for (a) $\theta=0$, (b) $\theta=0.3$, (c) $\theta=1$}
\label{fig.test1.indicators.time}
\end{center}
\end{figure}

The graphs of Fig. \ref{fig.test1.indicators} show the trend of the functions $\theta\mapsto I_1^{m,\,M,\,\mu}$, $\theta\mapsto I_2^{m,\,M,\,\mu}$, and $\theta\mapsto I_G^{m,\,M,\,\mu}$. Notice that, due to Eq. \eqref{eq.mu.nondim}, the moments of inertia of the multiscale mass are linear interpolations of the corresponding moments of inertia of the microscopic and the macroscopic masses. The latter are therefore also plotted in the graphs for reference. The most relevant fact is that the multiscale moments of inertia are almost constant with respect to $\theta$ (aside from small border effects, especially about $\theta=1$), which indicates that the rearrangement of the mass is basically the same at all scales. Therefore the microscopic and the macroscopic dynamics arising from pedestrian interactions \emph{are compatible with each other and make it possible a coupled approach by scale interpolation}.

The graphs of Fig. \ref{fig.test1.indicators.time} show the trend of the mappings $t\mapsto I_G^{m,\,M,\,\mu}$ for the three values of $\theta$ used in Fig. \ref{fig.test1.snapshots}, namely $\theta=0$ (Fig. \ref{fig.test1.indicators.time}a), $\theta=0.3$ (Fig. \ref{fig.test1.indicators.time}b), and $\theta=1$ (Fig. \ref{fig.test1.indicators.time}c). As pointed out in the Remark at the end of Section \ref{sect.num.approx}, the microscopic and the macroscopic scale always coexist and exchange information. In particular, by comparing Figs. \ref{fig.test1.indicators.time}a and \ref{fig.test1.indicators.time}c it can be noticed how \emph{the scale coupling realized in the model produces coherent results at both scales even when the dynamics is fully ruled by either of them only}. No significant qualitative and quantitative differences are observed in both the fully macroscopic ($\theta=0$, Fig. \ref{fig.test1.indicators.time}a) and the fully microscopic ($\theta=1$, Fig. \ref{fig.test1.indicators.time}c) case, meaning that there is no detachment of the two evolutions even when only one of them is actually the leading one. If this might be quite classical for the (Lagrangian) evolution of microscopic particles driven by a macroscopic flow, we stress that it is definitely by far less classical and obvious for the (Eulerian) evolution of a macroscopic flow driven by microscopic particles.

\subsection*{Test 2: Average outflow time}
In this test we address the case of a crowd leaving a room through a door in normal (\ie, no panic) conditions, and we investigate the influence of the coupled microscopic and macroscopic effects on the estimated average outflow time. This will provide meaningful insights into the way in which the microscopic granularity works within the macroscopic flow. The scenario of the simulation is depicted in Fig. \ref{fig.test2}a for the parameters listed in Table \ref{tab.parametri}. The (dimensionless) door width is $0.5$.

Let $\hat{\Omega}:=[0,\,3]\times [0,\,4]$ be the room that pedestrians are leaving. We consider the following average outflow time:
\begin{equation*}
	T_\ave^\mu:=\frac{1}{\mu_0(\hat{\Omega})}\lint_0^T t\mathcal{F}(t)\,dt,
\end{equation*}
where $\mathcal{F}(t)$ is the integral flux through the door (taken positive when outgoing) at time $t$. The final time $T>0$ is chosen so large that the room is definitely empty, \ie, $\mu_T(\hat{\Omega})=0$. Considering that $\mathcal{F}(t)=-\frac{d}{dt}\mu_t(\hat{\Omega})$ and using integration-by-parts, $T_\ave^\mu$ can be given the following form:
\begin{equation}
	T_\ave^\mu=\frac{1}{\mu_0(\hat{\Omega})}\lint_0^T\mu_t(\hat{\Omega})\,dt,
	\label{eq.test2.avetime}
\end{equation}
hence it is actually an outflow time weighted by the percent mass of crowd that, at each time instant, still has to leave the room.

The graphs of Figs. \ref{fig.test2}b, c show the trend of the function $\theta\mapsto T_\ave^\mu$ for a small crowd of $10$ pedestrians and a large crowd of $100$ pedestrians. In both cases, the two further curves $\theta\mapsto T_\ave^m$, $\theta\mapsto T_\ave^M$, computed by replacing $\mu_t$ in Eq. \eqref{eq.test2.avetime} with either $m_t$ or $M_t$, are plotted for suitable reference. Again, due to Eq. \eqref{eq.mu.nondim}, the multiscale average outflow time \eqref{eq.test2.avetime} is the linear interpolation of the corresponding microscopic and macroscopic times via the parameter $\theta$.

\begin{figure}[!t]
\begin{center}
\begin{minipage}[b]{0.3\textwidth}
	\centering
	\includegraphics[width=\textwidth]{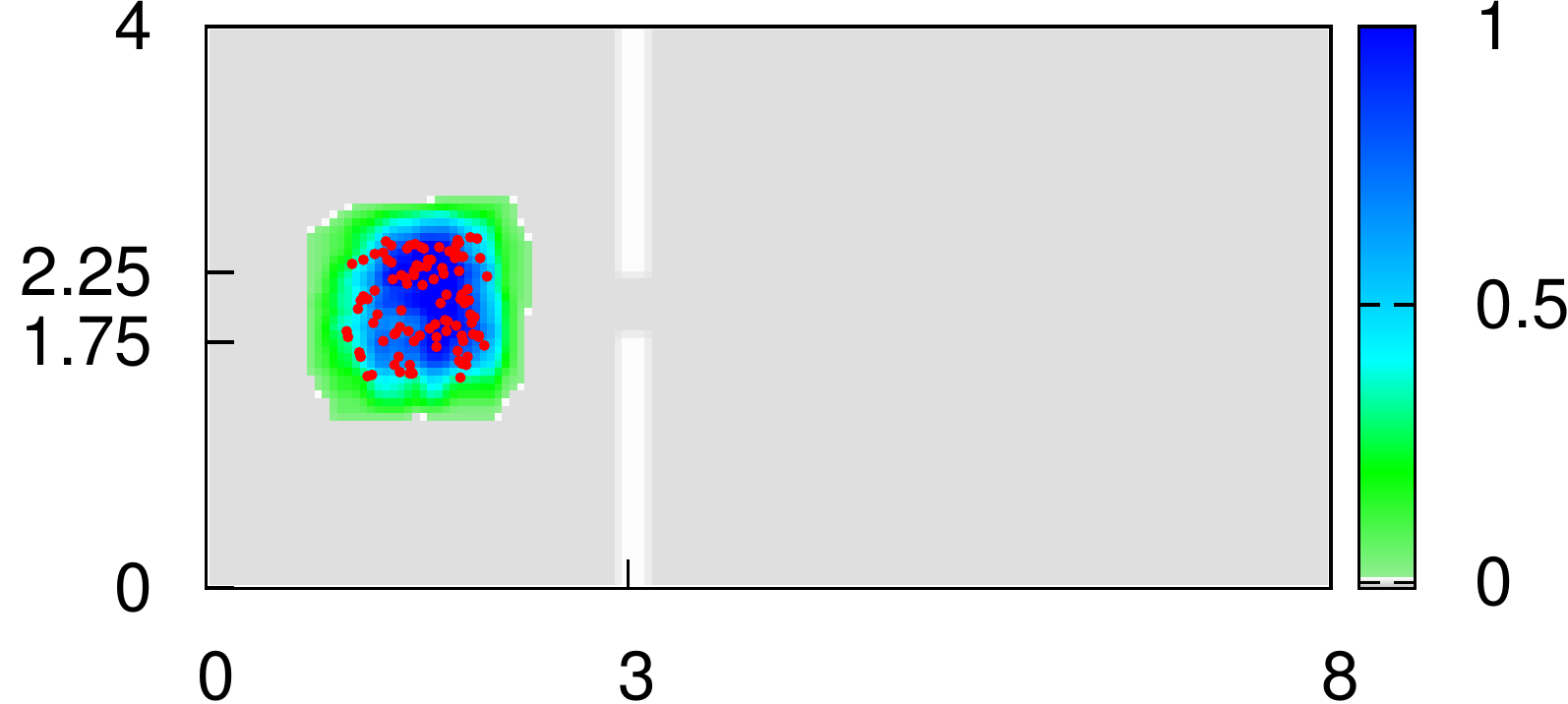} \\[0.3cm]
	\includegraphics[width=\textwidth]{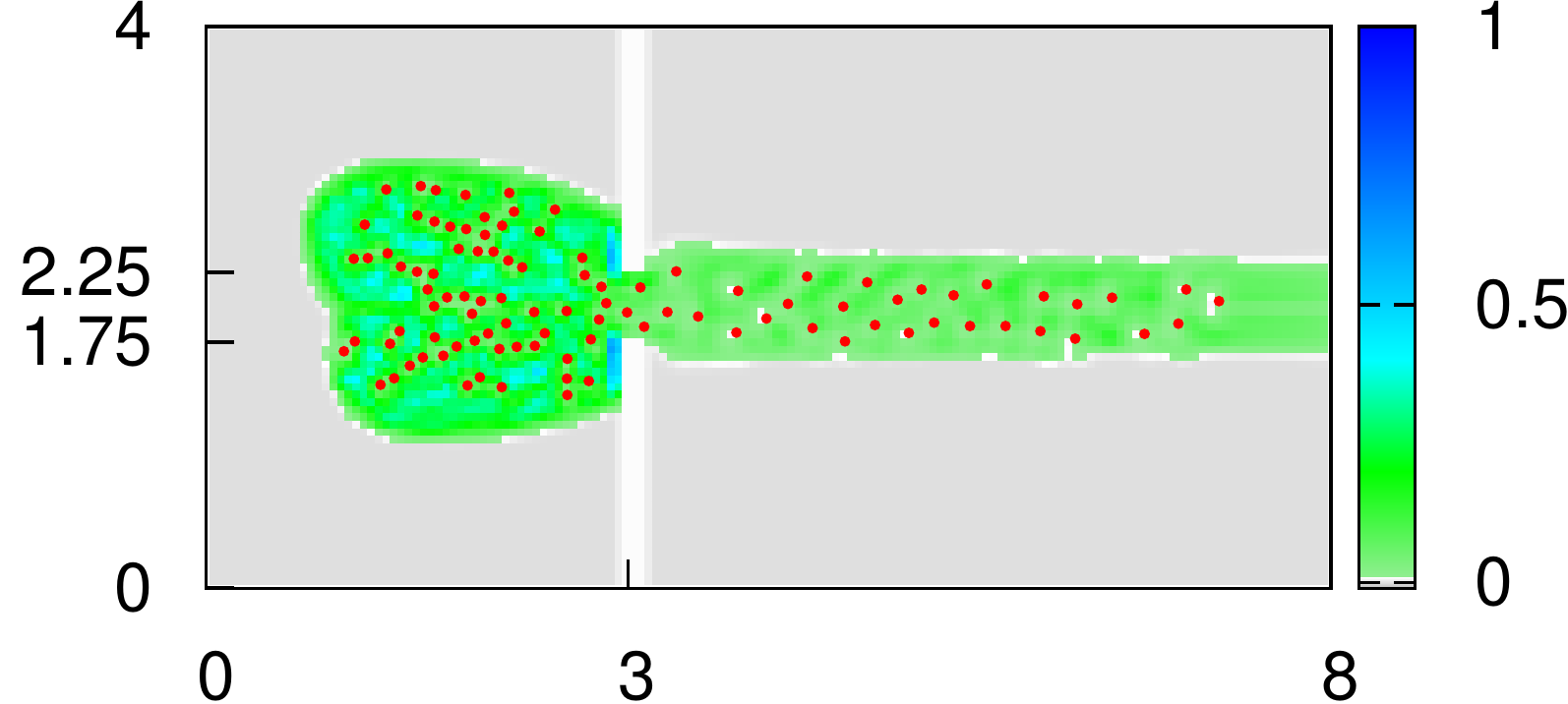} \\
	(a)
\end{minipage}
\quad
\begin{minipage}[b]{0.3\textwidth}
	\centering
	\includegraphics[width=\textwidth]{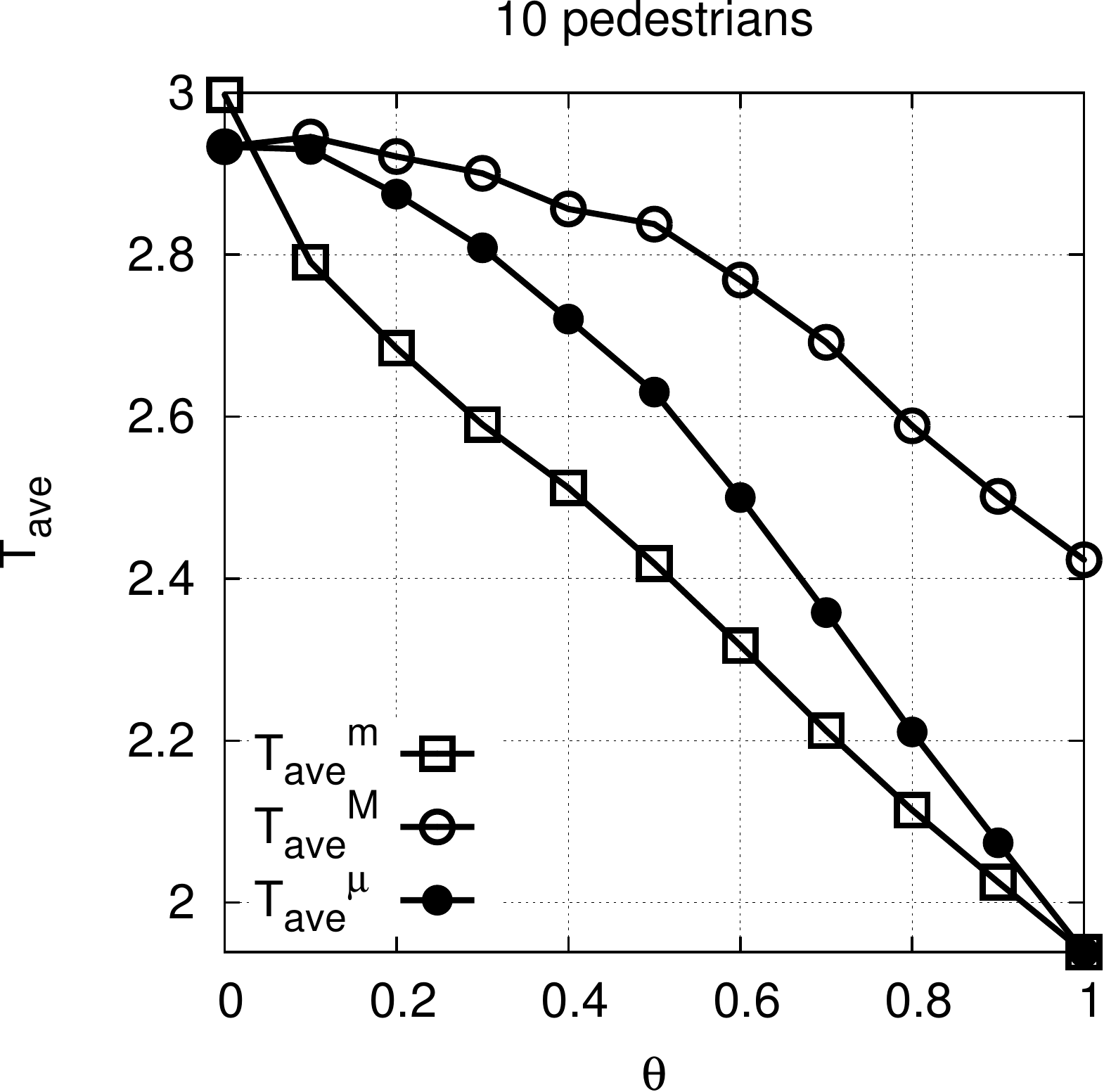} \\
	(b)
\end{minipage}
\quad
\begin{minipage}[b]{0.3\textwidth}
	\centering
	\includegraphics[width=0.96\textwidth]{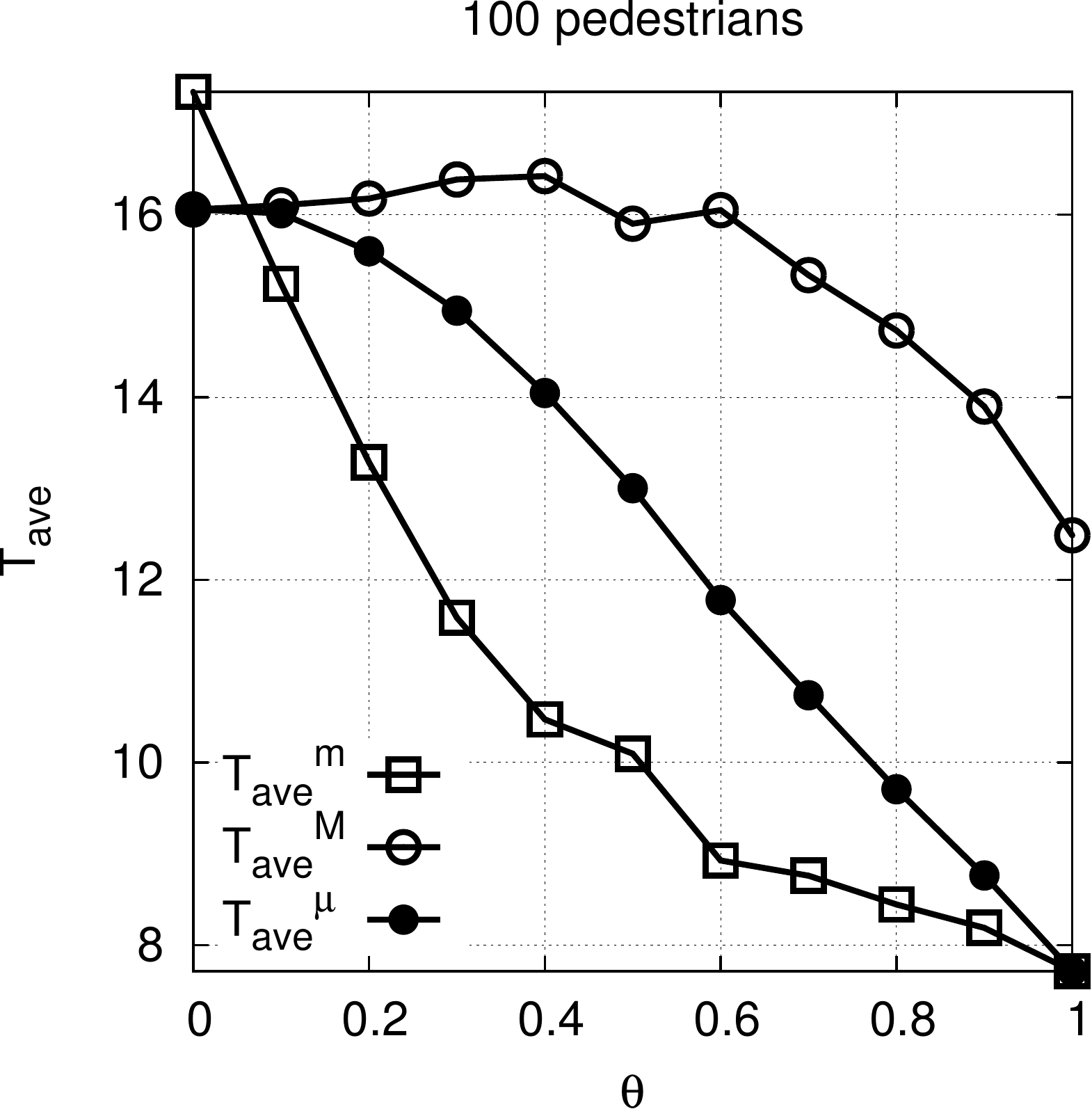} \\
	(c)
\end{minipage}
\caption{Test 2. (a) A crowd leaving a room through a door, initial condition (top) and underway outflow (bottom). (b) Average outflow time as a function of $\theta$ for a crowd of $10$ pedestrians and (c) $100$ pedestrians}
\label{fig.test2}
\end{center}
\end{figure}

The trend of the $T_\ave$'s is qualitatively similar for both the small and the large crowd, in particular it is decreasing with $\theta$. This elucidates the role played by a more and more influential microscopic granularity within the macroscopic flow: the more the multiscale coupling is biased toward the microscopic scale, the more fluent the crowd stream becomes (and consequently the average outflow time decreases). This is justifiable considering that $\theta$ can be viewed as the percent mass shifted from the macroscopic to the microscopic scale in consequence of the multiscale coupling. Subtracting interacting macroscopic mass from the system progressively reduces the action of the macroscopic interactions while enhancing that of the microscopic ones. Since the latter are less distributed in space, because the microscopic mass is clustered in point singularities, this ultimately results in fewer deviations from the desired velocity and the desired paths.

\subsection*{Test 3: Pedestrian flow through a bottleneck}
In this test we investigate the ability of the multiscale model to reproduce several flow conditions occurring when two groups of pedestrians, walking toward one another, share a narrow passage (e.g., a door).

From the modeling point of view it is necessary to handle two interacting populations of walkers, which will be done via two mass measures $\mu_t^p$, $p=1,\,2$, each obeying Eq. \eqref{eq.cons.mass.strong}. Either population has its own desired velocity $v_\des^p$ and interacts with the opposite population through the interaction velocity, which now depends on both the $\mu_t^p$'s. Specifically, denoting by $p^\ast$ the conjugate index\footnote{That is, $p^\ast=2$ if $p=1$ and $p^\ast=1$ if $p=2$.} of population $p$, we set
\begin{equation}
	\nu^p[\mu_t^p,\,\mu_t^{p^\ast}]=(1-\Theta)\nu^p[\mu_t^p]+\Theta\nu^{pp^\ast}[\mu_t^{p^\ast}],
	\label{eq.nu.2pop}
\end{equation}
where:
\begin{itemize}
\item $\nu^p[\mu_t^p]$ is the \emph{endogenous interaction}, \ie, the interaction with pedestrians of one's own population;
\item $\nu^{pp^\ast}[\mu_t^{p^\ast}]$ is the \emph{exogenous interaction}, \ie, the interaction with pedestrians of the opposite population;
\item $\Theta\in[0,\,1]$ is a dimensionless number fixing the strength of the exogenous against the endogenous interaction.
\end{itemize}
Both the endogenous and the exogenous interaction velocities are formally computed as in Eq. \eqref{eq.nu}, except that the exogenous one is integrated with respect to $\mu_t^{p^\ast}$. In addition, the exogenous interaction radius $R^{pp^\ast}$ need not be the same as the endogenous one $R^p$ if interactions with opposite walkers require more promptness than interactions with group mates\footnote{This simply corresponds to the function $f$ having different supports in the expressions of $\nu^p$ and $\nu^{pp^\ast}$.}.

For this test we let $\Theta=0.65$, thus $65\%$ repulsion is exogenous and $35\%$ is endogenous. Repulsion radii are $R_r^p=0.2$, $R_r^{pp^\ast}=0.35$, to be compared with the unit width of the narrow passage. Other relevant parameters are listed in Table \ref{tab.parametri}. The setting of the problem is displayed in the snapshots of Figs. \ref{fig.test3.theta0}, \ref{fig.test3.theta1}, \ref{fig.test3.theta03.snapshots}, in particular the blue crowd with red microscopic pedestrians, say population $1$, walks rightward whereas the yellow one with green microscopic pedestrians, say population $2$, walks leftward.

\begin{figure}[!t]
\begin{center}
\begin{minipage}[b]{0.6\textwidth}
	\centering
	\includegraphics[width=0.8\textwidth]{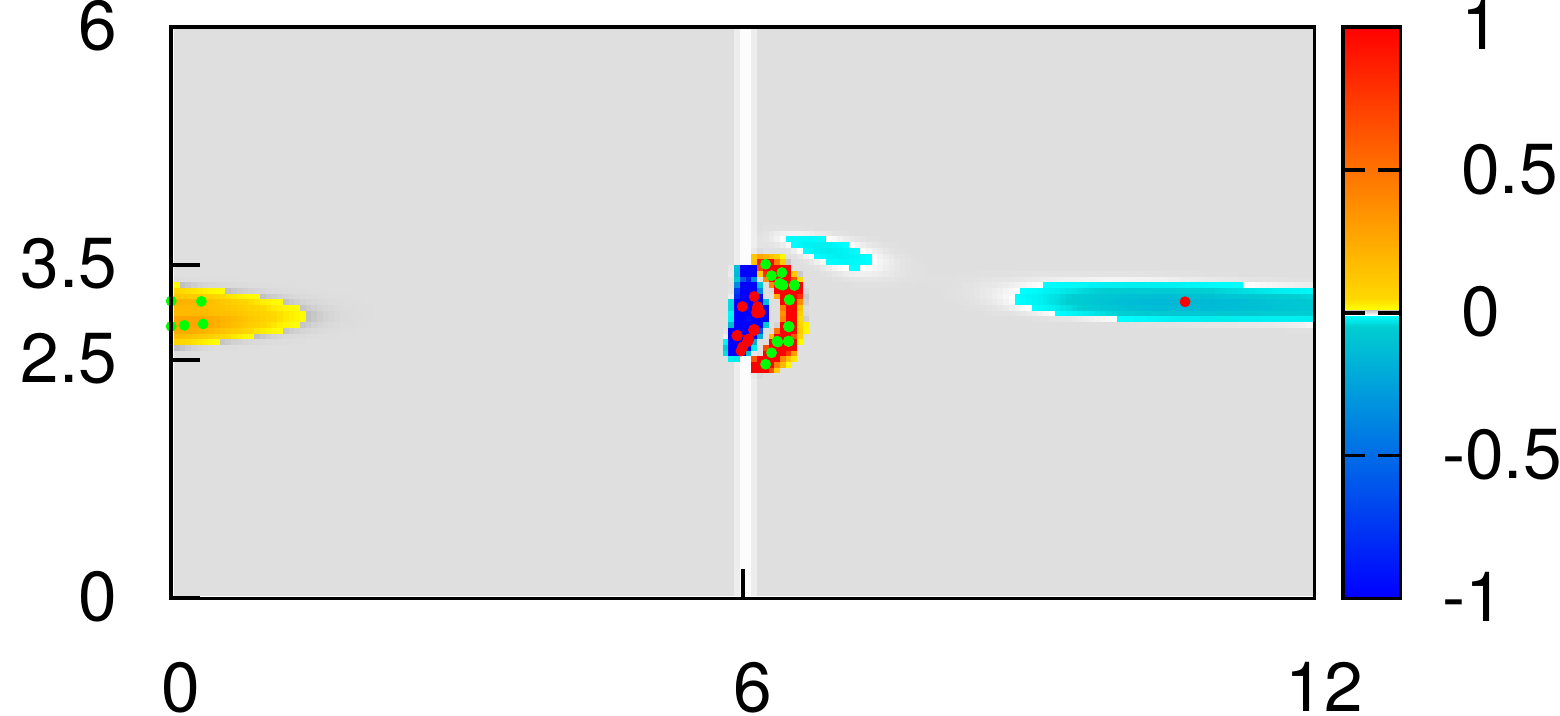} \\
	(a)
\end{minipage}
\begin{minipage}[b]{0.3\textwidth}
	\centering
	\includegraphics[width=\textwidth]{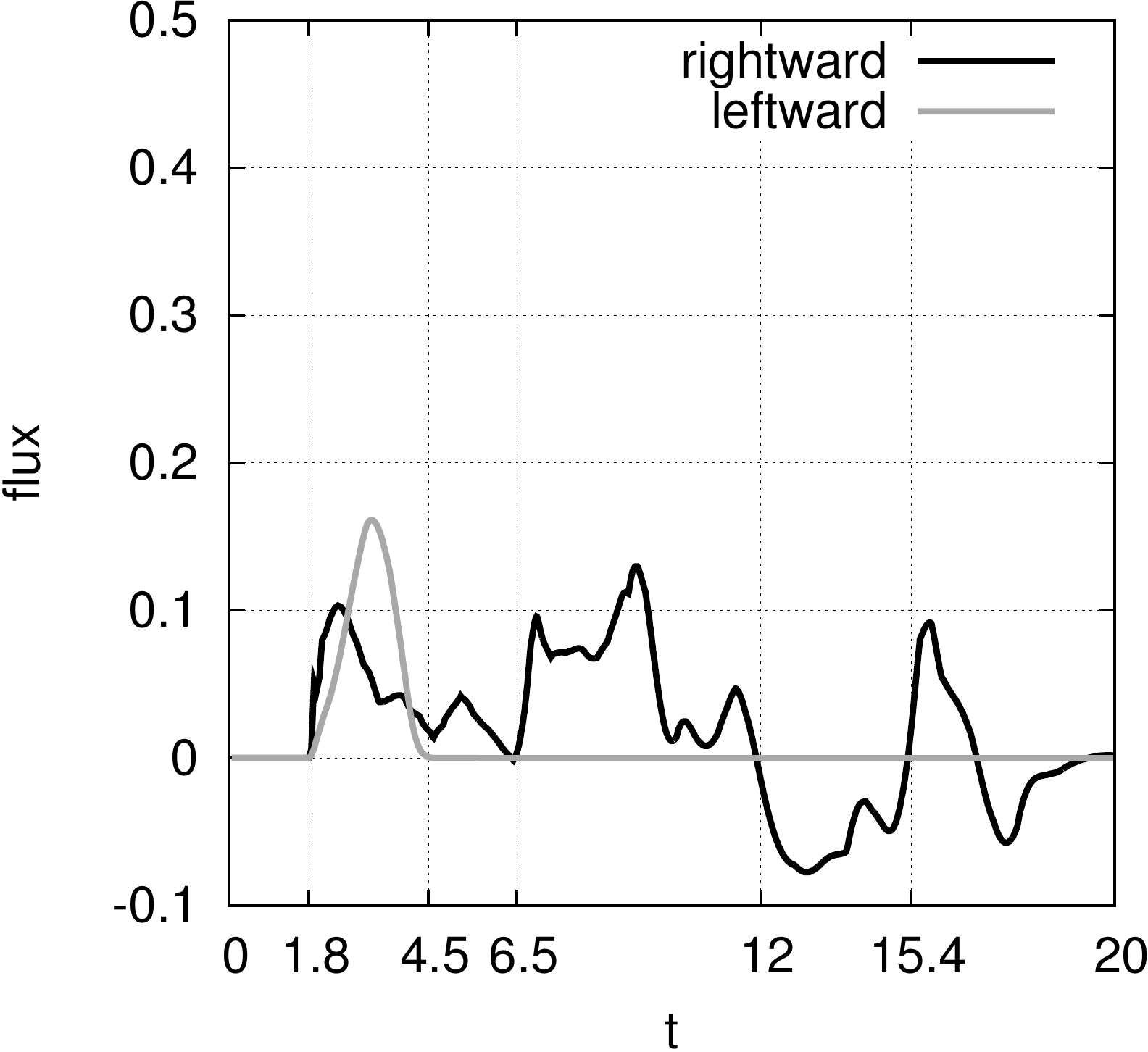} \\
	(b)
\end{minipage}
\caption{Test 3. Clogging at the bottleneck (left) and corresponding macroscopic fluxes (right) arising with the fully macroscopic dynamics}
\label{fig.test3.theta0}
\end{center}
\end{figure}

\begin{figure}[!t]
\begin{center}
\begin{minipage}[c]{0.3\textwidth}
	\centering
	\includegraphics[width=\textwidth]{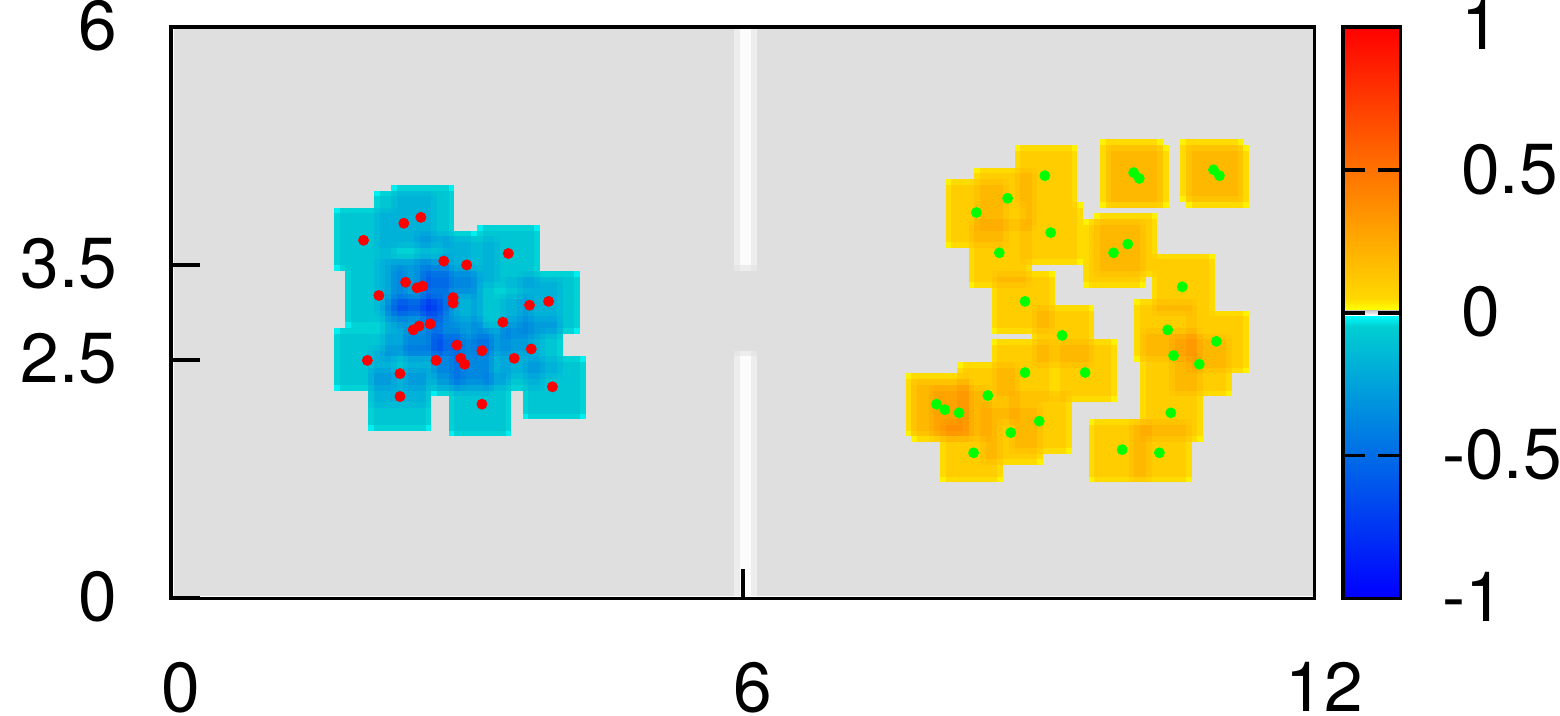} \\
	(a)
\end{minipage}
\begin{minipage}[c]{0.3\textwidth}
	\centering
	\includegraphics[width=\textwidth]{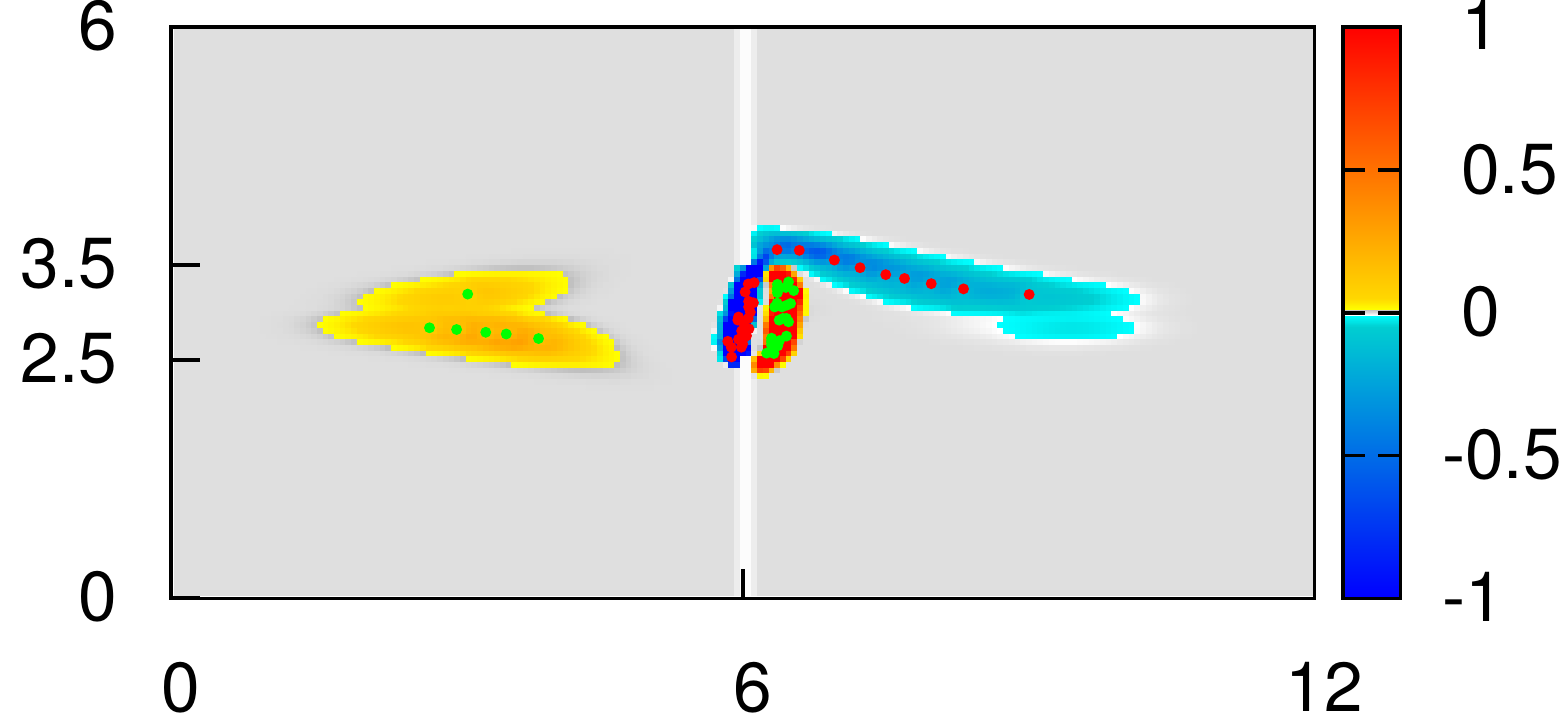} \\
	(b)
\end{minipage}
\begin{minipage}[c]{0.3\textwidth}
	\centering
	\includegraphics[width=\textwidth]{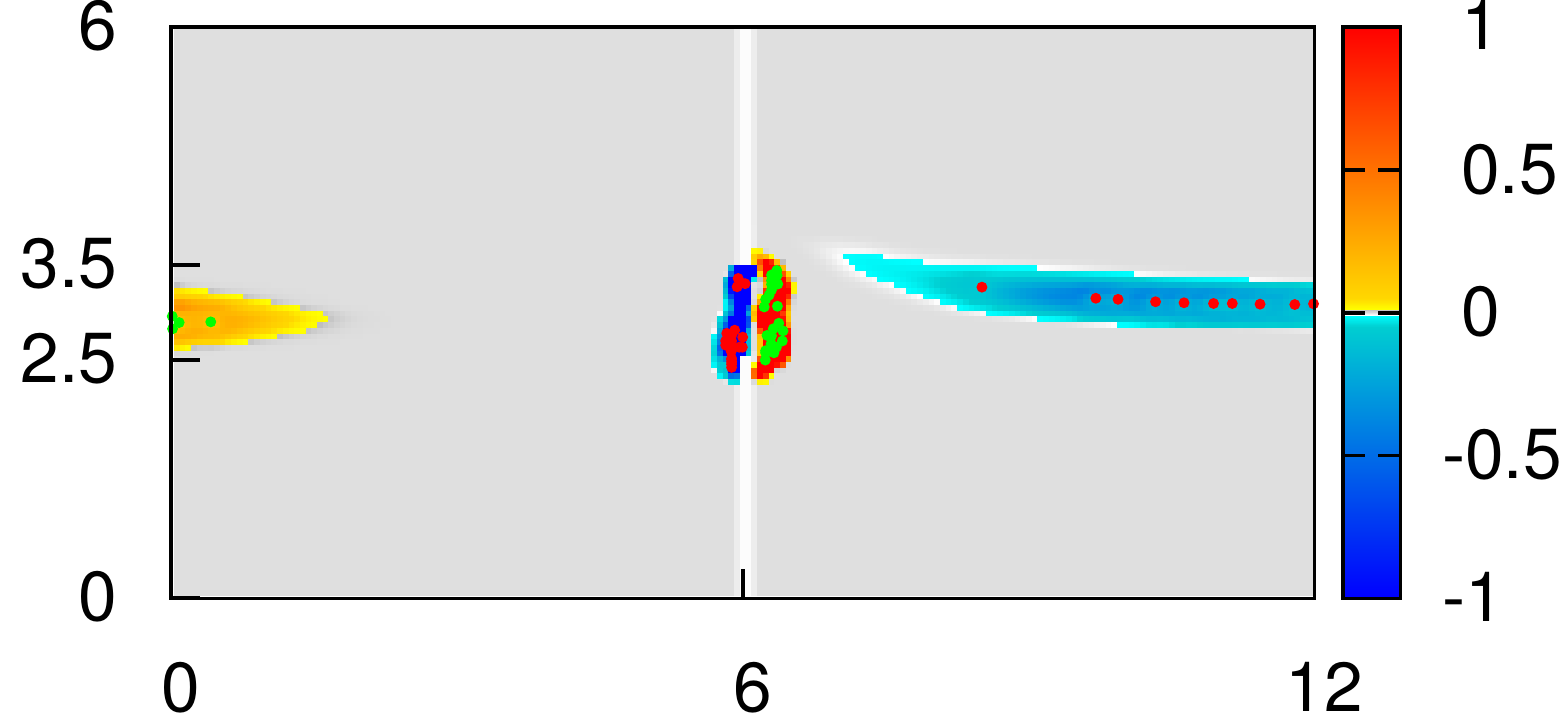} \\
	(c)
\end{minipage} \\[0.3cm]
\begin{minipage}[c]{0.3\textwidth}
	\centering
	\includegraphics[width=\textwidth]{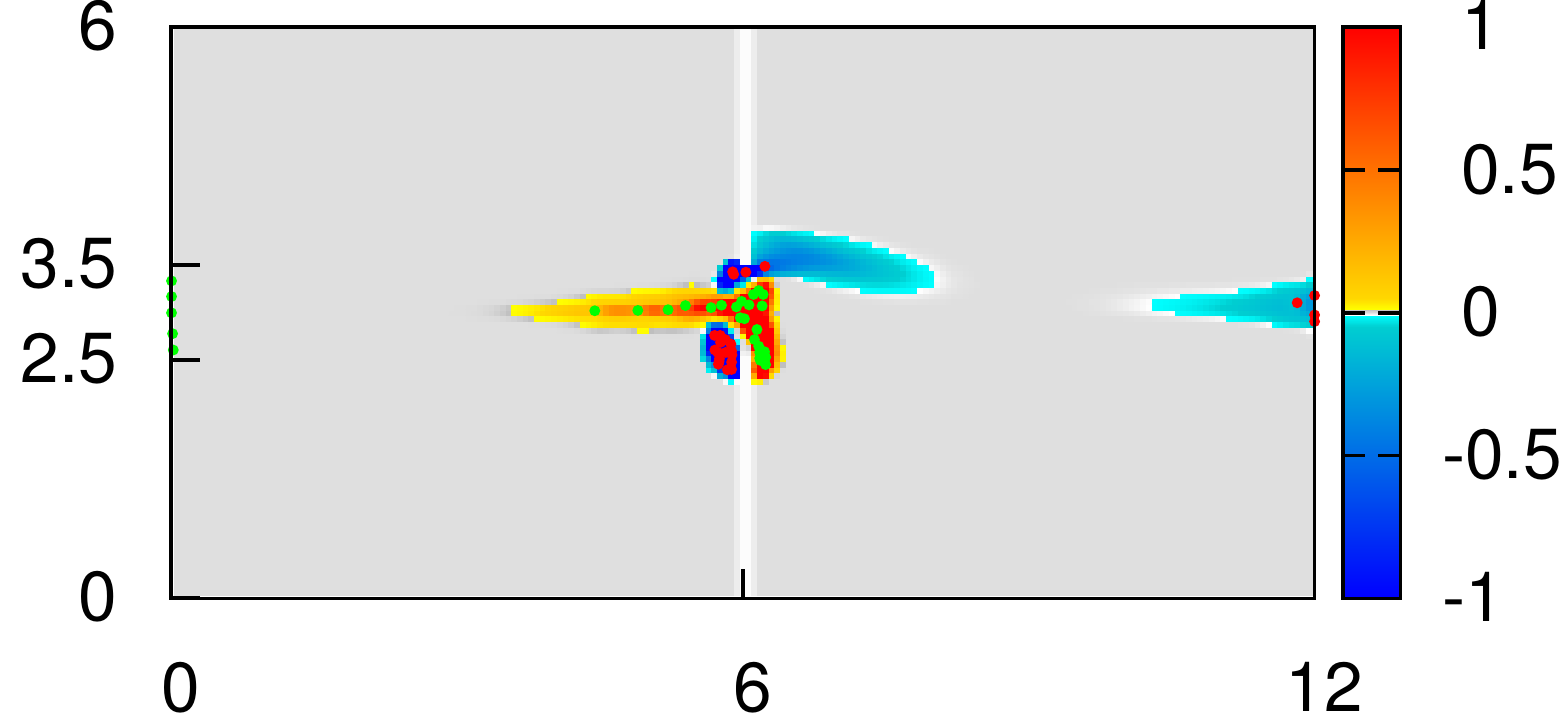} \\
	(d)
\end{minipage}
\begin{minipage}[c]{0.3\textwidth}
	\centering
	\includegraphics[width=\textwidth]{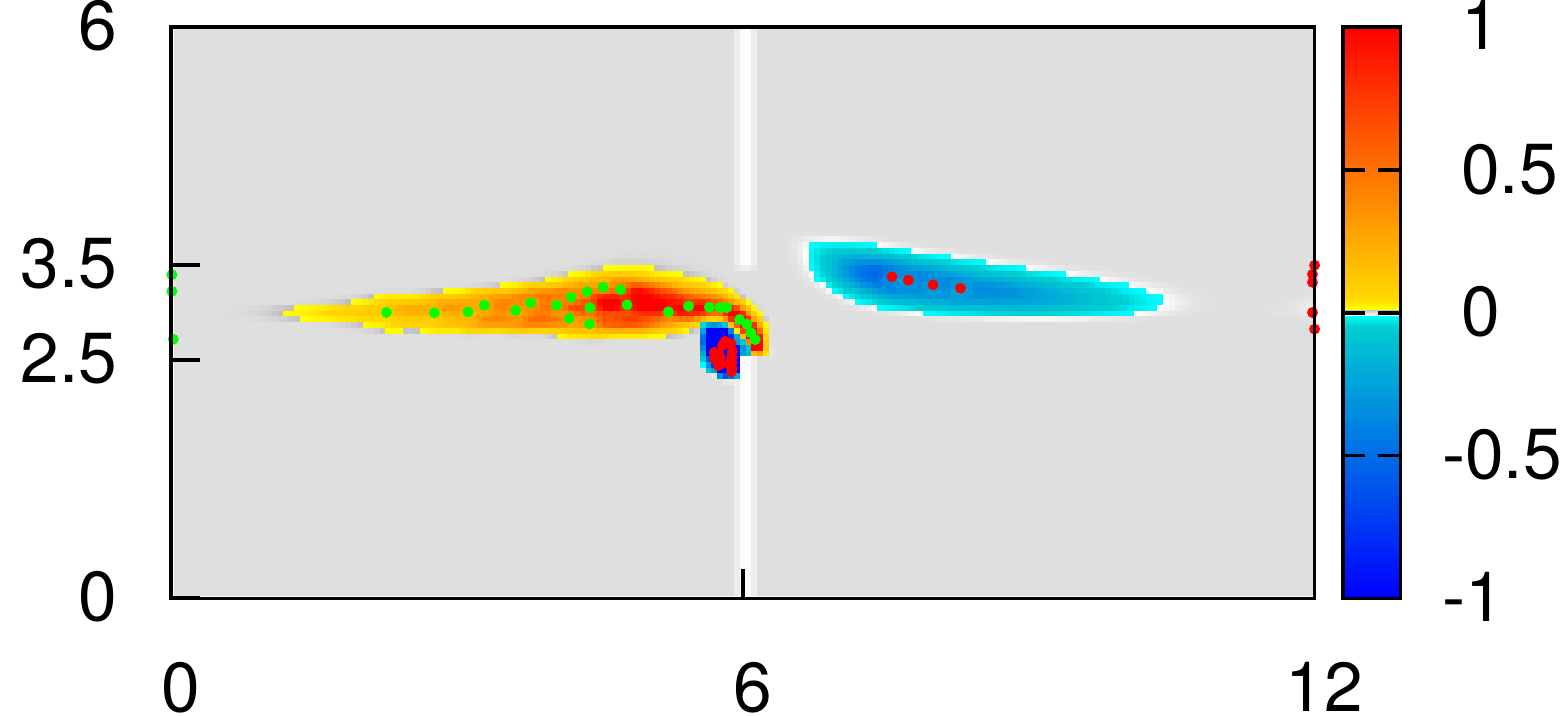} \\
	(e)
\end{minipage}
\begin{minipage}[c]{0.3\textwidth}
	\centering
	\includegraphics[width=\textwidth]{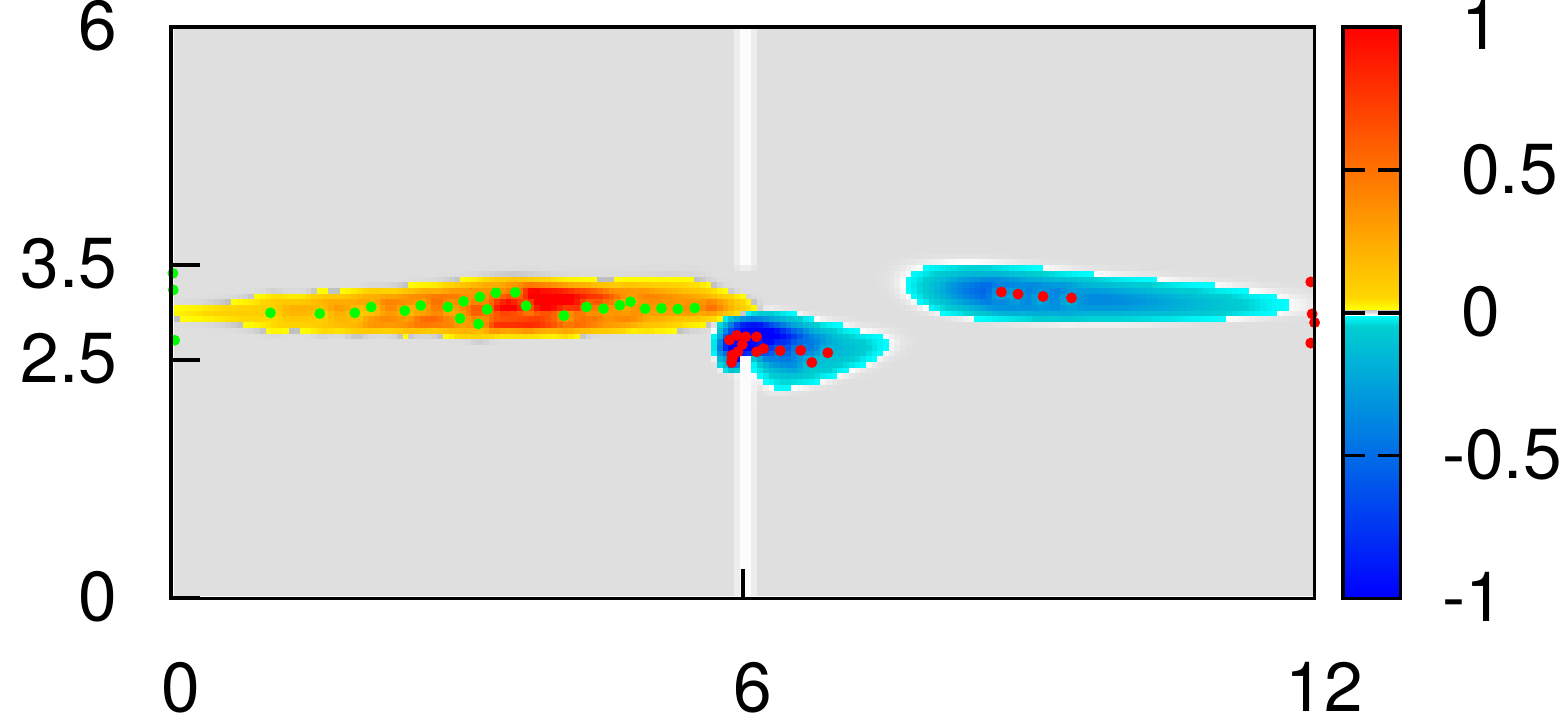} \\
	(f)
\end{minipage}
\caption{Test 3. Alternate flows at the bottleneck in the multiscale model ($\theta=0.3$). Negative values of the density of the population walking rightward are for pictorial purposes only}
\label{fig.test3.theta03.snapshots}
\end{center}
\end{figure}
\begin{figure}[!t]
\begin{center}
\begin{minipage}[c]{0.3\textwidth}
	\centering
	\includegraphics[width=\textwidth]{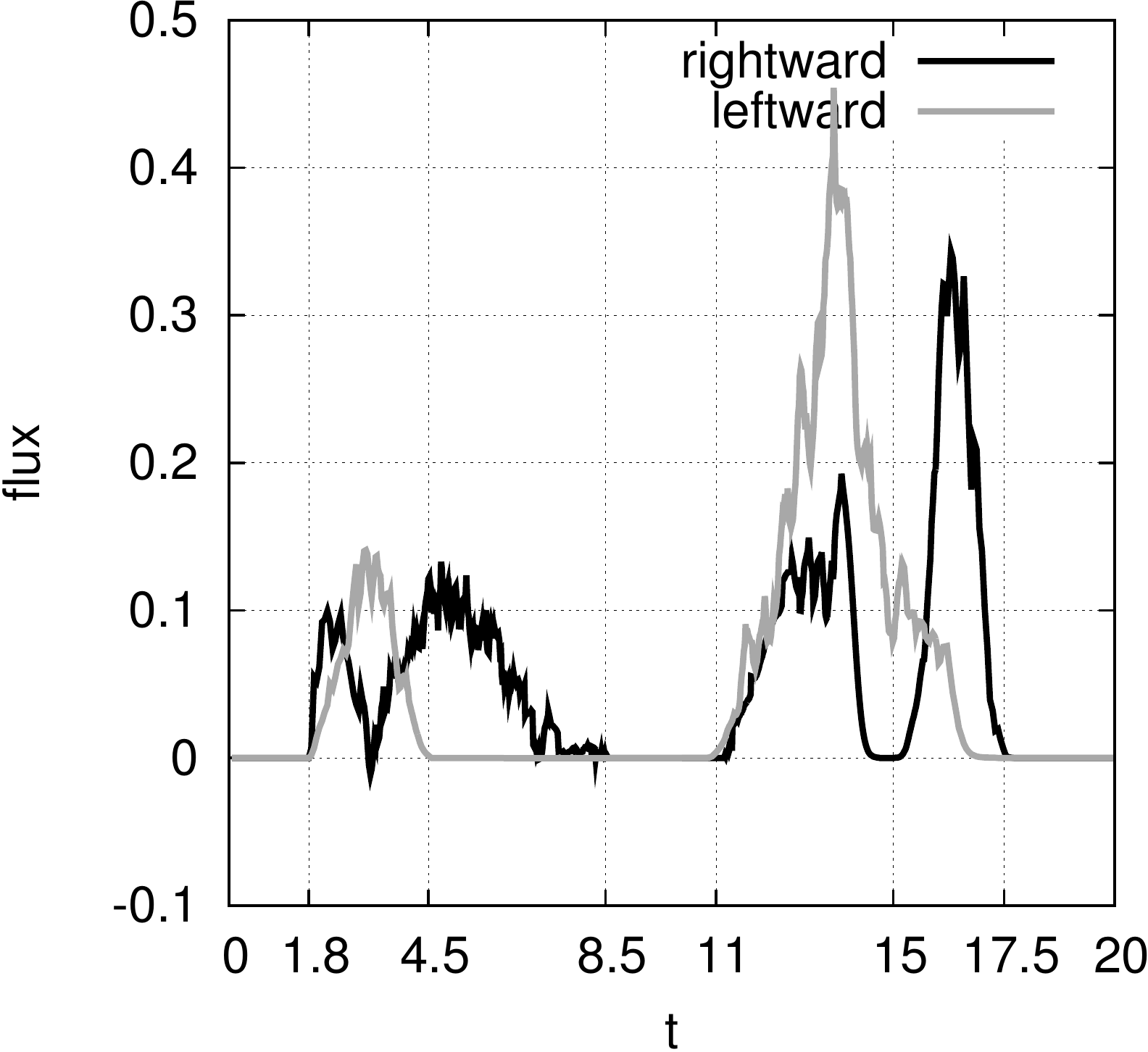} \\
	(a)
\end{minipage}
\hspace{0.3cm}
\begin{minipage}[c]{0.3\textwidth}
	\centering
	\includegraphics[width=0.95\textwidth]{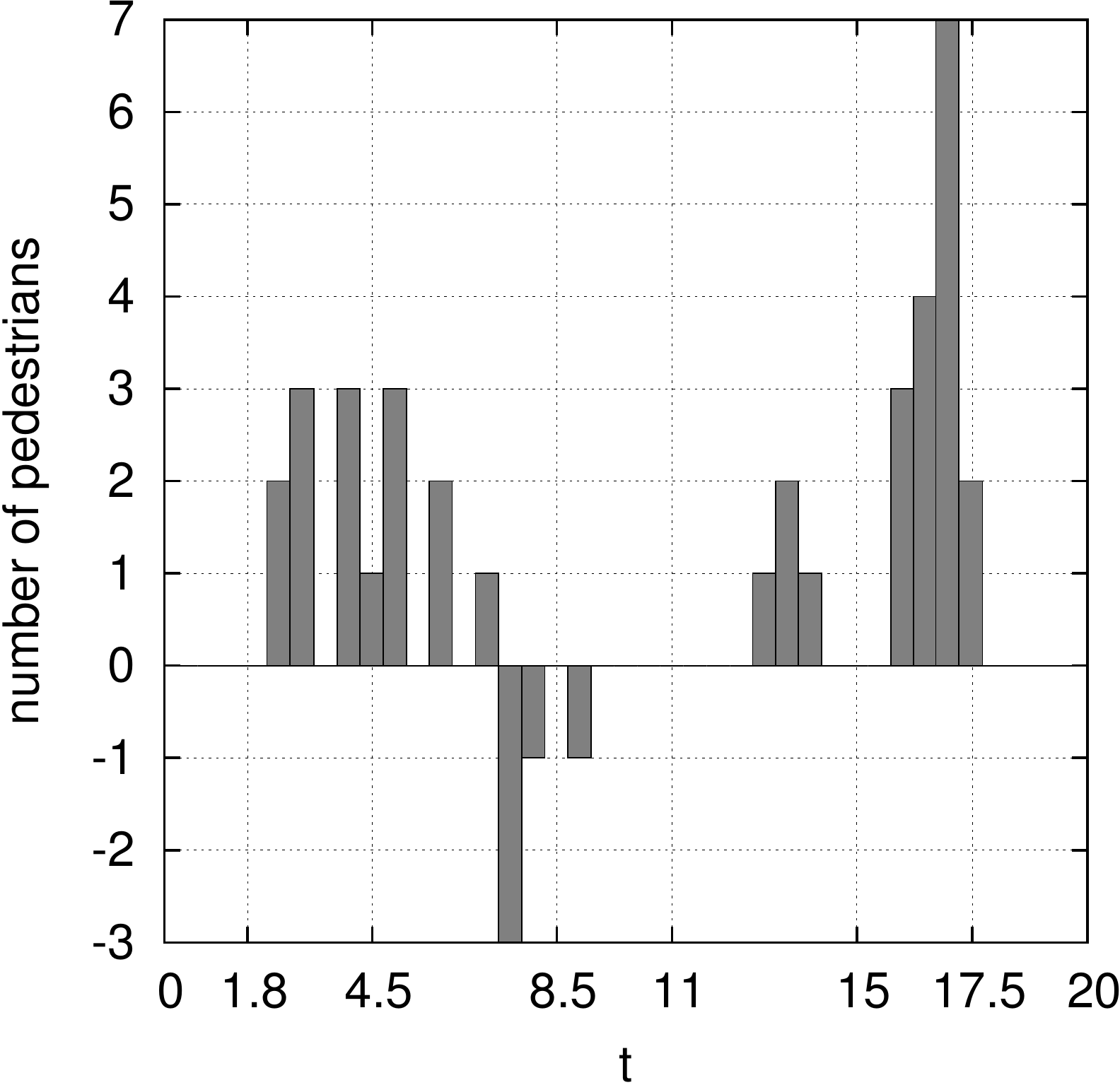} \\
	(b)
\end{minipage}
\hspace{0.3cm}
\begin{minipage}[c]{0.3\textwidth}
	\centering
	\includegraphics[width=0.95\textwidth]{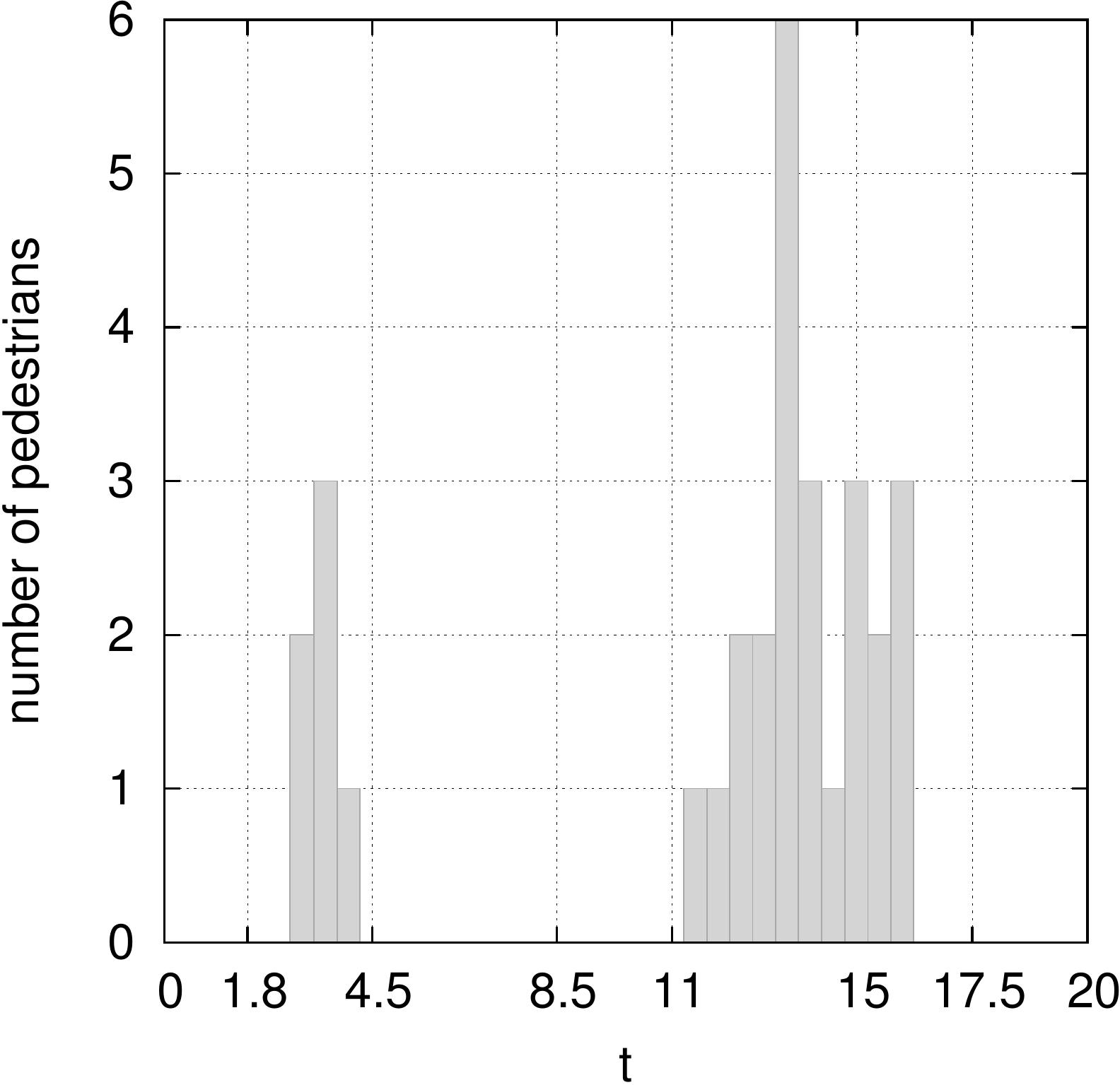} \\
	(c)
\end{minipage}
\caption{Test 3. Multiscale model ($\theta=0.3$): (a) macroscopic fluxes across the bottleneck, (b)-(c) number of microscopic pedestrians crossing the bottleneck rightward and leftward, respectively}
\label{fig.test3.theta03.flux}
\end{center}
\end{figure}

\begin{figure}
\begin{center}
\begin{minipage}[b]{0.6\textwidth}
	\centering
	\includegraphics[width=0.8\textwidth]{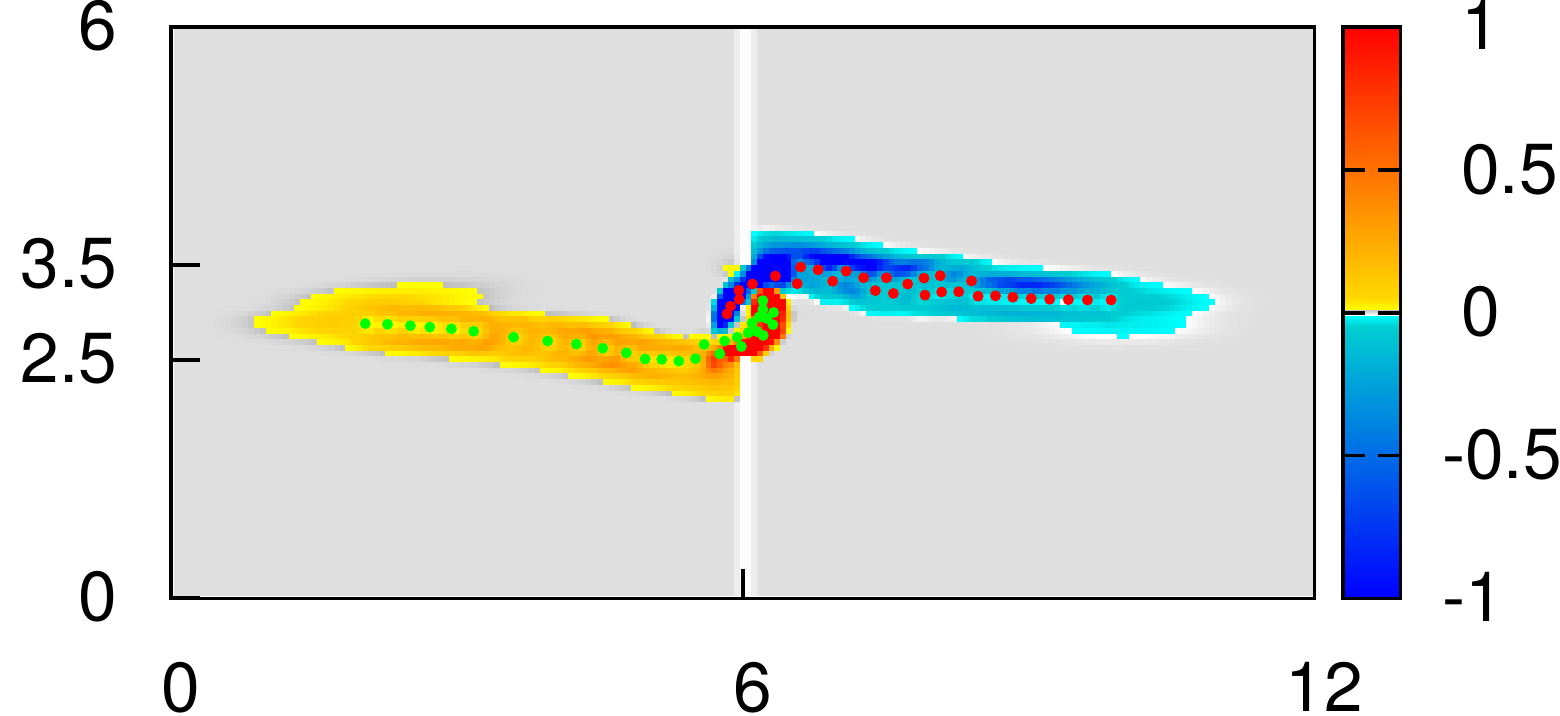} \\
	(a)
\end{minipage}
\begin{minipage}[b]{0.3\textwidth}
	\centering
	\includegraphics[width=\textwidth]{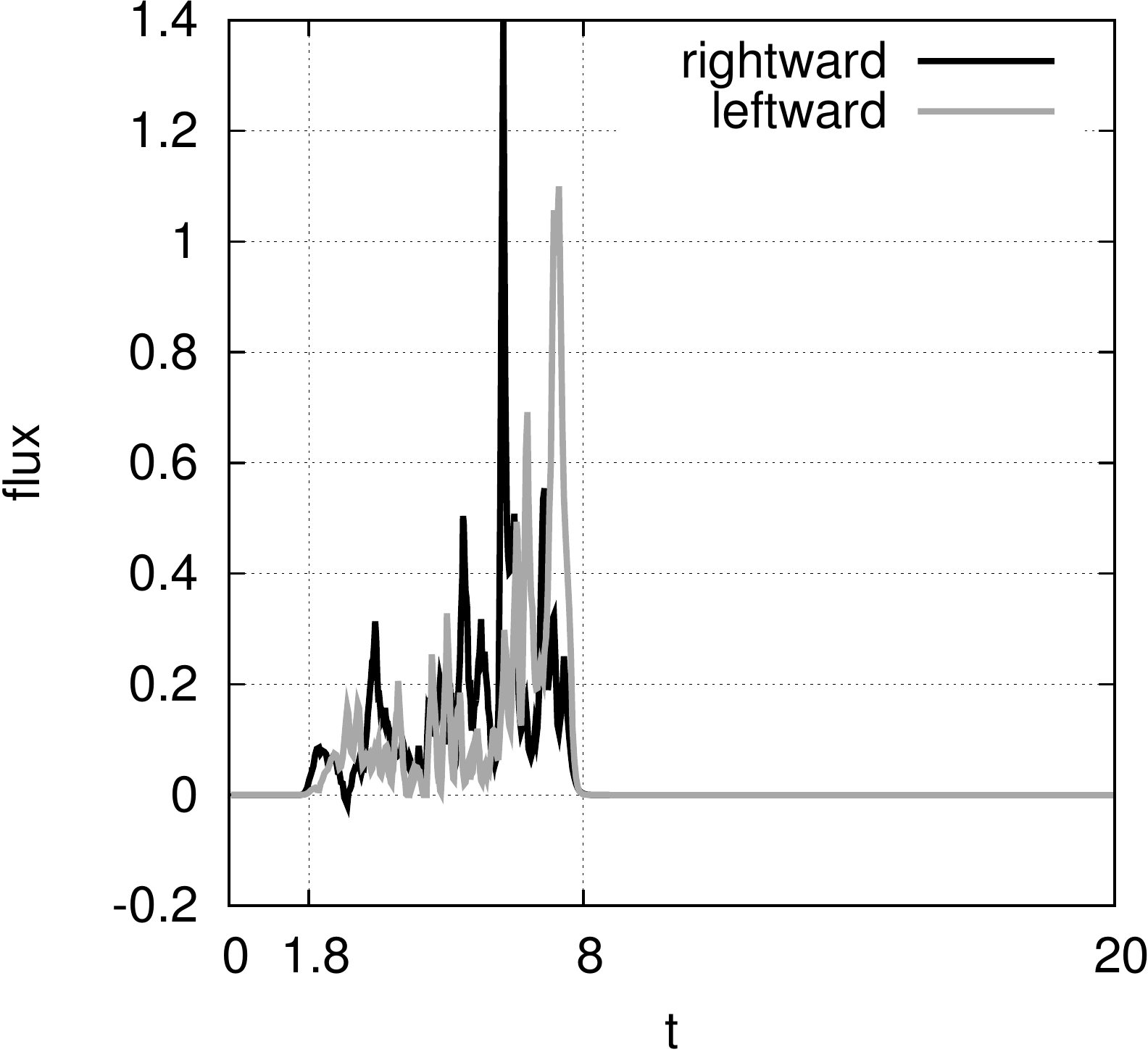} \\
	(b)
\end{minipage}
\caption{Test 3. Alternate lanes through the bottleneck (left) and corresponding macroscopic fluxes (right) emerging with the fully microscopic dynamics. Notice that both fluxes are identically zero for $t\geq 8$ because by then both populations have completely flowed across the door}
\label{fig.test3.theta1}
\end{center}
\end{figure}

Let us begin from the case $\theta=0$, with the macroscopic scale leading the dynamics. The bottleneck tends to clog (Fig. \ref{fig.test3.theta0}a): no density nor microscopic pedestrians flow through, except for a small mass passing initially when the passage is still free. This is well confirmed by the time trend of the macroscopic flux across the bottleneck (Fig. \ref{fig.test3.theta0}b): that of population $2$ is permanently zero for $t\geq 4.5$, whereas that of population $1$ oscillates between small positive and negative values, which implies that population $1$ is pushed backward by population $2$ as soon as it tries to cross. 

By increasing $\theta$ to an intermediate value between $0$ and $1$, a multiscale coupling is realized. For $\theta=0.3$, the resulting dynamics is depicted in Fig. \ref{fig.test3.theta03.snapshots} and summarized in Fig. \ref{fig.test3.theta03.flux} by the time trend of the macroscopic and microscopic fluxes across the bottleneck. The model reproduces now the oscillations of the passing direction at the bottleneck described e.g., in \cite{helbing2001trs,helbing1995sfm,helbing2001sop}. In more detail, starting from the initial condition depicted in Fig. \ref{fig.test3.theta03.snapshots}a, pedestrians of population $2$ are induced to stop at the bottleneck while those of population $1$ go through at one side (Fig. \ref{fig.test3.theta03.snapshots}b, Fig. \ref{fig.test3.theta03.flux} for $4.5\leq t\leq 8.5$). After some time, population $2$ reorganizes and stops the flow of population $1$ (Fig. \ref{fig.test3.theta03.snapshots}c, Fig. \ref{fig.test3.theta03.flux} for $8.5\leq t\leq 11$), then its larger mass stuck at the bottleneck gives it locally the necessary strength for repelling opposite walkers and gaining room in the middle (Fig. \ref{fig.test3.theta03.snapshots}d, Fig. \ref{fig.test3.theta03.flux} for $11\leq t\leq 15$). Some walkers of population $1$ remain trapped by the stream of population $2$ and cannot access the passage (Fig. \ref{fig.test3.theta03.snapshots}e, Fig. \ref{fig.test3.theta03.flux} for $t\approx 15$) until  most of population $2$ has flowed through (Fig. \ref{fig.test3.theta03.snapshots}f, Fig. \ref{fig.test3.theta03.flux} for $15\leq t\leq 17.5$).

From the modeling point of view, the difference with the case $\theta=0$ is that the multiscale coupling shifts some macroscopic mass ($30\%$ in this example) onto the microscopic pedestrians, all other parameters and initial conditions being unchanged. The inhomogeneous distribution of this microscopic mass induces a break of symmetry between the interfacing populations, which finally leads to an alternate occupancy of the passage according to the repulsion prevailing locally in space and time.

Setting $\theta=1$, which amounts to shifting the whole mass onto the microscopic pedestrians and having the microscopic scale lead the dynamics, produces instead the outcome displayed in Fig. \ref{fig.test3.theta1}. Now the microscopic granularity fully dominates, hence the stream is the most fluent one (cf. also Test 2). As a result, the bottleneck interferes less with the stream than in the previous cases, and the model reproduces the alternate oppositely walking lanes (Fig. \ref{fig.test3.theta1}a) extensively observed as one of the main effects of self-organization in real crowds (cf. \cite{helbing2001sop,hoogendoorn2005sop} and references therein). The time trend of the macroscopic flux across the bottleneck (Fig. \ref{fig.test3.theta1}b) confirms that the two populations flow simultaneously through the passage, with comparable fluxes, in the interval $1.8\leq t\leq 8$. After the time $t=8$ the macroscopic fluxes are identically zero because by then both populations have completely flowed across the door.

\begin{figure}[!t]
\begin{center}
\begin{minipage}[c]{0.47\textwidth}
	\centering
	\includegraphics[width=\textwidth]{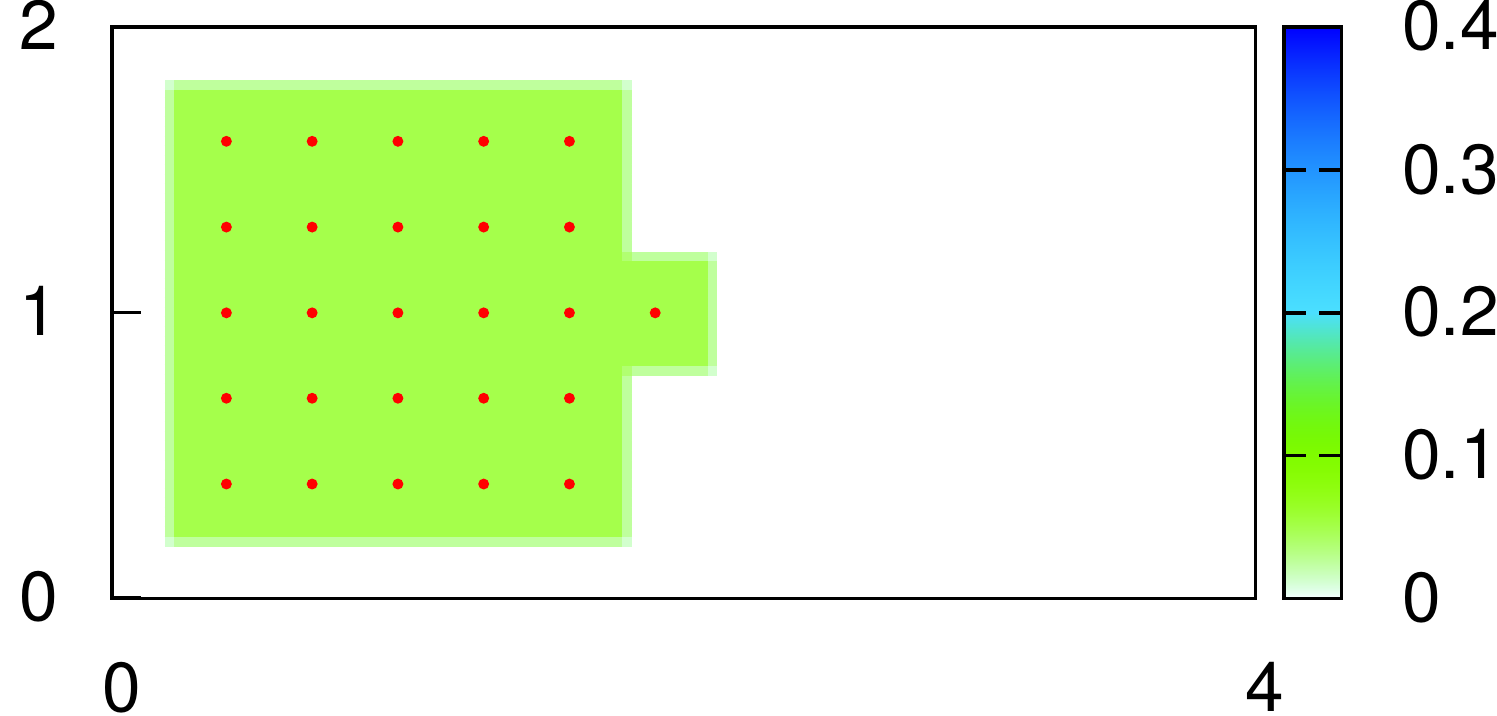} \\
	(a)
\end{minipage}
\hspace{0.5cm}
\begin{minipage}[c]{0.47\textwidth}
	\centering
	\includegraphics[width=\textwidth]{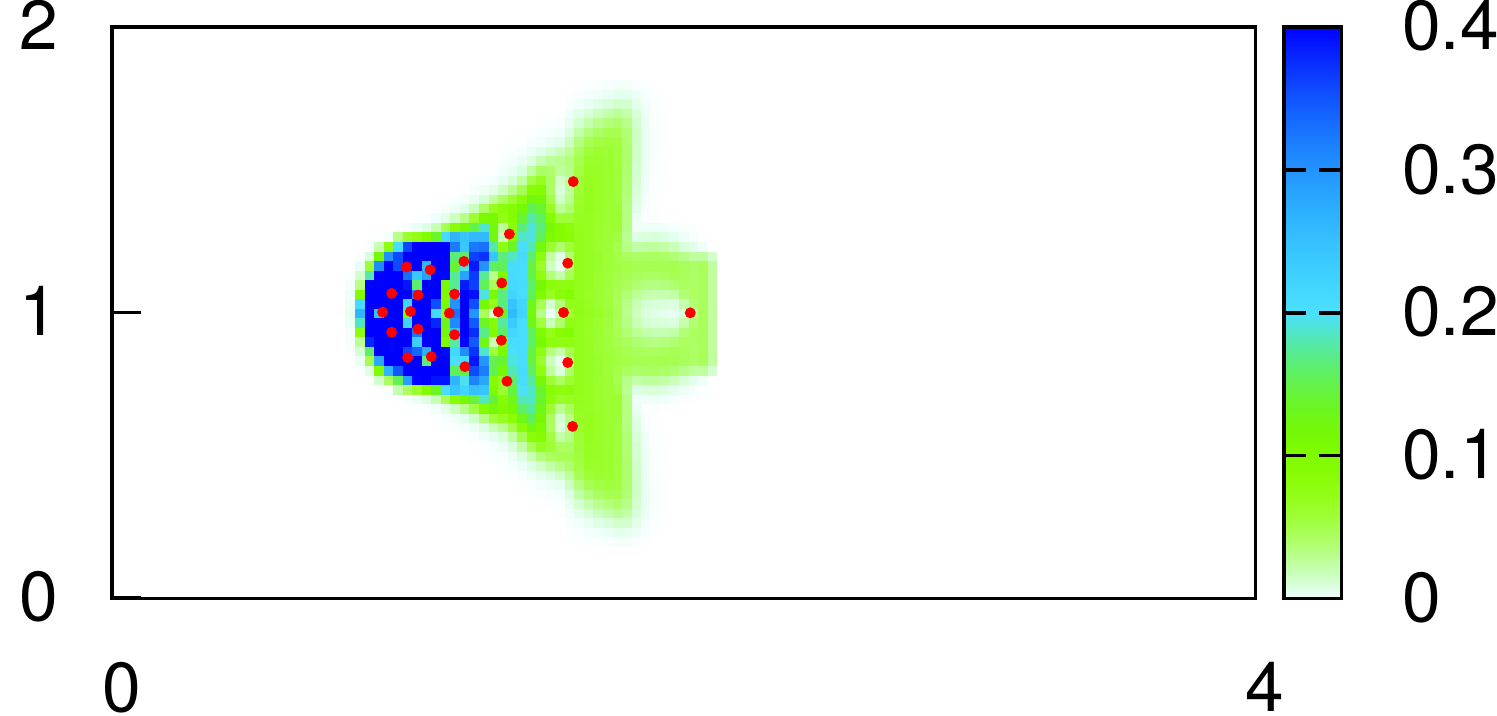} \\
	(b)
\end{minipage} \\[0.3cm]
\begin{minipage}[c]{0.47\textwidth}
	\centering
	\includegraphics[width=\textwidth]{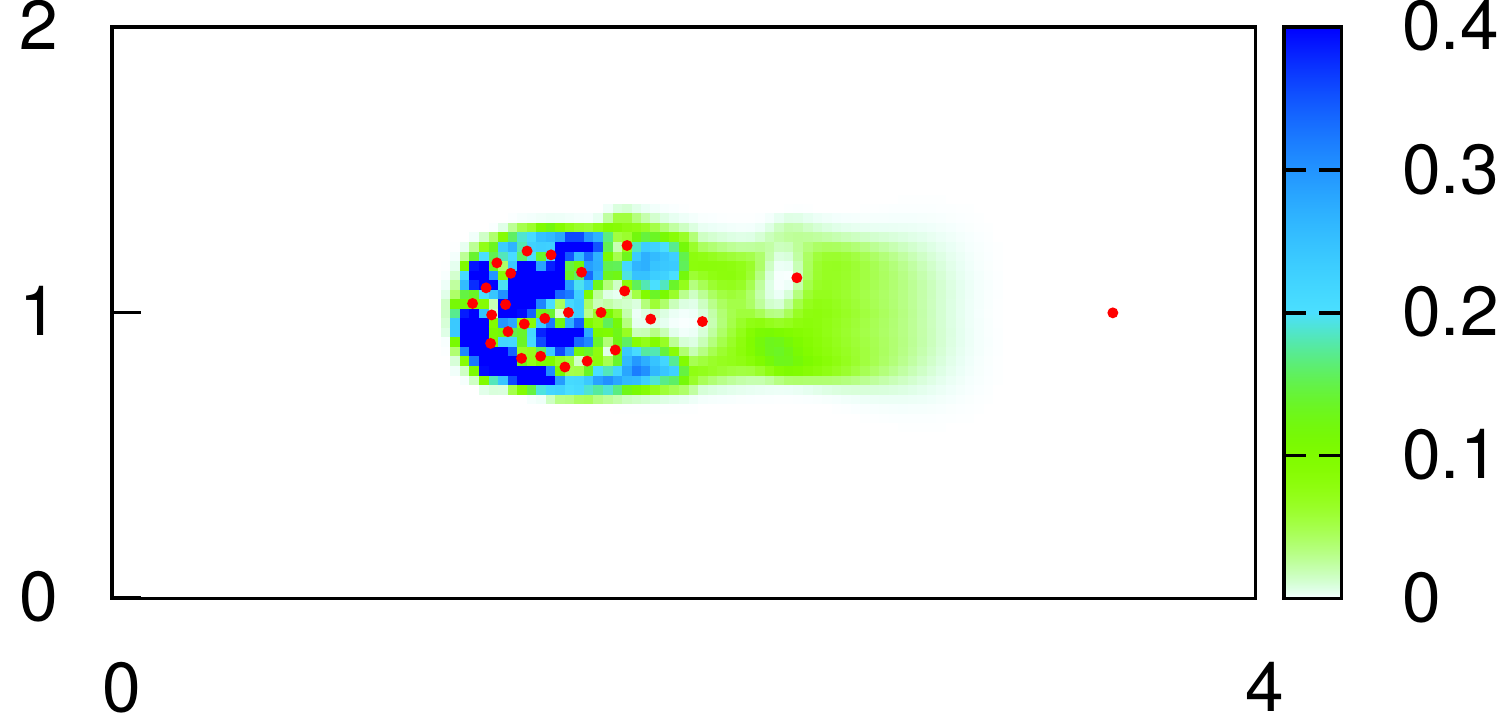} \\
	(c)
\end{minipage}
\hspace{0.5cm}
\begin{minipage}[c]{0.47\textwidth}
	\centering
	\includegraphics[width=\textwidth]{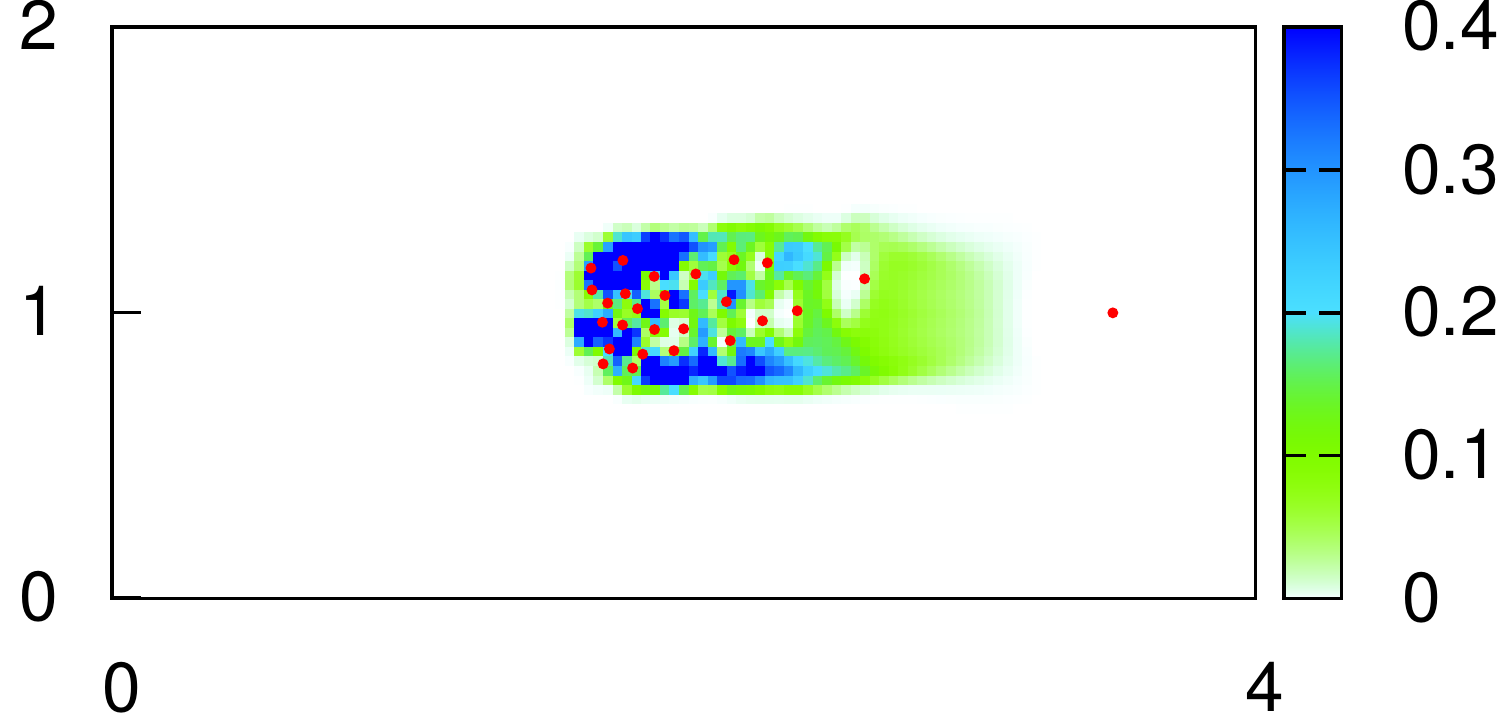} \\
	(d)
\end{minipage}
\caption{Test 4. (a) initial condition, (b) the group assumes an elongated configuration, (c) the group is formed and follows the leader, (d) the leader waits for the group while the group moves on}
\label{fig.test4.snapshots}
\end{center}
\end{figure}

\subsection*{Test 4: Macroscopic effect of a microscopic leader}
In this last test we outline a capability of the multiscale model, not yet highlighted so far, which will surely deserve further investigation. More precisely, we are referring to the use of the microscopic scale for modeling some features of the system which could not be described in a purely macroscopic framework, but which nonetheless affect the macroscopic dynamics. This is essentially different from the previous tests, where the same effects, such as repulsion and obstacle avoidance, were described at both scales. The main novelty here is that the system includes a microscopic term \emph{with no macroscopic counterpart}. 

We consider the case of a crowd following a leader, for instance a group of tourists and their guide. The leader is a microscopic pedestrian who behaves in a different way with respect to all of the other group members: s/he is the only one \emph{informed} of the way to go, hence s/he walks with a pre-assigned velocity ($0.4\mathbf{i}$ in this example) independently of the others (\ie, s/he does not interact with the rest of the group). S/he only stops when her/his distance from the group becomes too large. The followers have zero desired velocity, because they are not informed of the way to go, and experience both frontal attraction and frontal repulsion with their group mates, including the leader. (Like in Test 1, the angle $\alpha_{xy}$ in Eq. \eqref{eq.nu.tdiscr} is computed by assuming conventionally $v_\des=\mathbf{i}$). Attraction acts against group dispersion, and is needed especially in order for the crowd to follow the leader. Instead, repulsion is intended for collision avoidance among group mates. The radii $R_a$, $R_r$ are equal (cf. Table \ref{tab.parametri}), in particular $R_a$ is so small that the tail of group does not feel the leader ahead.

The group starts from the square-shaped distribution depicted in Fig. \ref{fig.test4.snapshots}a, with the leader in front. Then, after a transient (Fig. \ref{fig.test4.snapshots}b), it assumes a horizontally elongated shape (Fig. \ref{fig.test4.snapshots}c) as a result of joint attractive and repulsive effects. With no leader such a configuration would be an equilibrium, as attraction and repulsion balance. However, as soon as the leader starts moving forward undisturbed, pedestrians at the front, who can feel him directly, are attracted and move forward in turn. At the same time, pedestrians at the rear are attracted toward group mates in front. This makes the information on the way to go travel backward across the group, which ultimately moves forward as a whole.

It is worth stressing again that at the macroscopic scale there is no counterpart of the microscopic leader. This implies that the macroscopic interaction velocity $\nu[M_t]$ is not affected by the microscopic leader, therefore the macroscopic mass feels the latter only through the microscopic interaction velocity $\nu[m_t]$.

This test shows that our multiscale framework is suitable to reproduce a well known feature of self-organizing groups, namely the fact that a small number of informed agents can move the whole group in the desired direction \cite{couzin2005eld}. In particular, thanks to the multiscale coupling, this effect is appreciable also at the macroscopic scale, which would not be suitable by itself to model differences among the individuals.

\section{Conclusions and future research}
\label{sect.conclusions}
In this paper we have presented a measure-based multiscale method for modeling pedestrian flow. We point out that neither a purely microscopic ODE-based approach nor a purely macroscopic PDE-based approach is new in the literature for this kind of application. The novelty here is the way of coupling the two scales in a rigorous mathematical framework. This is possible thanks to the measure-theoretic approach, which makes no \emph{a priori} distinction between the scales, and to the fact that, by a proper scaling, the microscopic and the macroscopic model can reproduce the flow of the same mass of pedestrians with comparable outcomes (cf. the numerical test 1). We stress that introducing microscopic heterogeneity in a macroscopic model is not straightforward. Adding random disturbances to the macroscopic variables may lead to apparently good results but it cannot be mathematically nor physically justified in an averaged context. Instead, our method allows to add granularity to the macroscopic flow and to preserve at the same time physical meaning and mathematical rigor.

From the modeling side, it is worth noticing that a macroscopic model of pedestrian flow is useful to get overall distributed information, especially in connection with design, control, and optimization issues. However, as demonstrated by our numerical simulations, a certain amount of granularity is often crucial to catch some aspects of self-organization in crowds triggered by the microscopic inhomogeneities of the flow (cf. the numerical test 3). Of course, it has to be expected that the outcome at whatever scale partly depends on the tuning of the parameters of the model. Therefore, it is not our purpose to state that the multiscale approach is always better (\ie, more realistic) than either the microscopic or the macroscopic approach by itself. Rather we believe that the proposed technique offers a convenient way to make the two scales interact and jointly contribute to the final result.

The present form of our multiscale approach is mainly concerned with the same mass of pedestrians modeled at both the microscopic and the macroscopic scale. The multiscale coupling is then realized by scale interpolation. However, the numerical test 4 demonstrates that the framework is suitable also to model features at either scale, which have no explicit counterpart at the other scale and nonetheless affect crucially the overall dynamics. As a research development, we plan to further generalize our multiscale approach in this direction, having in mind specific applications related to traffic flow.

Pedestrians and cars share some relevant features, such as a desired velocity driving them toward specific destinations and a frontally restricted visual field. Actually car movements are much more constrained than pedestrians', hence self-organization is more limited, however not completely inhibited. For example, an application quite considered in the technical literature \cite{lan2005ica,lee2009nam} concerns mixed traffic conditions with few mopeds within a flow of cars.

\section*{Acknowledgments}
A. Tosin was funded by a post-doctoral research scholarship ``Compa\-gnia di San Paolo'' awarded by the National Institute for Advanced Mathematics ``F. Severi'' (INdAM, Italy).

\bibliographystyle{plain}
\bibliography{CePbTa-micromacro}

\end{document}